\documentclass[aps,pra,reprint,superscriptaddress]{revtex4-2}
\usepackage{graphicx}
\usepackage{mathpazo}
\usepackage{dcolumn}
\usepackage{xcolor}
\usepackage{bm}
\usepackage{braket}
\usepackage{lipsum}
\usepackage{amsmath, amsthm, amsfonts, amssymb, bbm, dsfont, quantikz, ragged2e} 
\usepackage{enumitem}
\usepackage{mathtools}
\usepackage{nicematrix}
\usepackage{pifont}
\usepackage[most]{tcolorbox}
\usepackage[dvipsnames]{xcolor}
\usepackage{nicematrix} 
\usepackage[colorlinks,linkcolor=blue,urlcolor=blue, citecolor=blue]{hyperref}

\usepackage{tikz}
\usetikzlibrary{quantikz2}

\usetikzlibrary{decorations}

\usepackage{algorithm}
\usepackage{algpseudocode}

\usepackage{centernot}

\usepackage{tikz}
\usetikzlibrary{quantikz2}

\newtheorem{theorem}{Result}
\newtheorem{lemma}{Lemma}

\newtheorem{corollary}{Corollary}

\newcommand{\ketbra}[1]{\vert #1 \rangle \langle #1 \vert}
\newcommand{\ketbratwo}[2]{\vert #1 \rangle \langle #2 \vert}
\newcommand{\onenorm}[1]{ \| #1 \|_1}

\definecolor{teal}{RGB}{42, 157, 143}
\definecolor{yellow}{RGB}{233, 196, 106}
\definecolor{red}{RGB}{210, 66, 66}
\definecolor{lred}{RGB}{222,94,100}
\definecolor{lblue}{RGB}{179, 235, 242 }

\newcommand{\fref}[1]{\textcolor{blue}{\hyperref[#1]{Fig.$\,$\bfseries\ref{#1}}}}
\newcommand{\lemref}[1]{\textcolor{blue}{\hyperref[#1]{Lemma$\,$\bfseries\ref{#1}}}}
\newcommand{\thmref}[1]{\textcolor{blue}{\hyperref[#1]{Thm.$\,$\bfseries\ref{#1}}}}

\usepackage{silence}
\newcommand{\drawCustomCircletwo}[8]{%
    \tikz[baseline=(circle_node.base)] {
        \draw[draw=#4, fill=#5, fill opacity=#6, line width=#7]
            (0,0) circle (#3);
        \node[circle, inner sep=0pt, outer sep=0pt] (circle_node) at (0,0) {#8};
    }%
}

\tikzset{
    every node/.style={font=\small},
    arrow/.style={-{Stealth}, thick},
    implies/.style={->, double equal sign distance, thick}
}

\newtcolorbox[auto counter]{mybox}[2][]{
    breakable = false,
    enhanced,
    sharp corners,
    colback=violet!3!white,
    colframe=violet!40!white,
    fonttitle=\bfseries,
    title={\centering \strut #2}, 
    enlarge bottom at break by=5mm,
    enlarge top at break by=5mm,
    overlay first={%
        \draw[black, line width=0.5mm](frame.south west)--(frame.south east);
        \node[anchor=north east] at (frame.south east) {continued on next page};
    },
    overlay middle={%
        \draw[black, line width=0.5mm](frame.south west)--(frame.south east);
        \draw[black, line width=0.5mm](frame.north west)--(frame.north east);
        \node[anchor=north east] at (frame.south east) {continued on next page};
        \node[anchor=south west] at (frame.north west) {continued from next page};
    },
    overlay last={%
        \draw[black, line width=0.5mm](frame.north west)--(frame.north east);
        \node[anchor=south west] at (frame.north west) {continued from last page};},
    #1
}

\let\oldaddcontentsline\addcontentsline
\renewcommand{\addcontentsline}[3]{}

\begin{document}

\preprint{APS/123-QED}

\title{Quantum State Discrimination With Stabilizer Circuits}

\author{Benjamin Stratton}
\email{benstratton6@live.co.uk}
\affiliation{International Institute of Physics, Federal University of Rio Grande do Norte, 59078-970, Natal, Brazil}

\author{Santiago Zamora}
\affiliation{International Institute of Physics, Federal University of Rio Grande do Norte, 59078-970, Natal, Brazil}
\affiliation{Departamento de Física Teórica e Experimental, Federal University of Rio Grande do Norte, 59078-970 Natal, Brazil}

\author{Rafael Chaves}
\affiliation{International Institute of Physics, Federal University of Rio Grande do Norte, 59078-970, Natal, Brazil}

\date{\today}

\begin{abstract}
The task of quantum state discrimination provides an operational characterization of distinguishability and plays a central role in quantum information science. Since the achievable success probability in quantum state discrimination depends both on the states being discriminated and on the allowed measurements, it is natural to study discrimination under physically motivated measurement constraints. Here, we investigate minimum-error quantum state discrimination under measurements implementable by both fixed and adaptive stabilizer circuits. Given arbitrary stabilizer-state ancillas, we show that fixed stabilizer circuits provide no additional discrimination power, whereas adaptive circuits do. Then, focusing on single-qubit systems, we derive analytical expressions for the success probability across both circuit classes when supplied with arbitrary (non-stabilizer) ancillas, showing that the performance gap between fixed and adaptive circuits persists.
We further consider discrimination with non-stabilizerness incorporated directly into the measurement operators. Here, the optimization is formulated as a semidefinite program, and analytical bounds interpolating between the stabilizer and Helstrom limits for qubits are derived. Finally, we illustrate applications in quantum random access codes and bounds on unitary synthesis fidelity with a finite number of magic states.
\end{abstract}

\maketitle

\section{Introduction}
The amount of classical information that can be extracted from a quantum system is fundamentally limited by the distinguishability of quantum states~\cite{Helstrom1976}; unlike classical states, nonorthogonal quantum states cannot be perfectly distinguished. 
This is one of the defining features of quantum mechanics, underpinning a broad range of quantum information tasks, including quantum communication~\cite{Bae_2015}, quantum cryptography~\cite{Bennett1984, Pirandola2018}, quantum hypothesis testing~\cite{Audenaert2008}, quantum sensing~\cite{zhuang2017optimum}, and quantum channel discrimination~\cite{Zhuang2020}. 

It is often useful to formalize state distinguishability as a task, known as quantum state discrimination (QSD)~\cite{Bae_2015}. In the simplest scenario, a referee prepares one of two states, $\rho$ or $\sigma$, with equal probability and sends them to a player. The player must then identify which state they were given. The optimal success probability in the minimum-error setting is given by the Helstrom bound~\cite{Helstrom1976},
\begin{equation}
    P^*_{\rm suc}(\Delta) = \frac{1}{2} + \frac{1}{4}\|\Delta\|_1,
\end{equation}
where $\Delta=\rho-\sigma$. The optimal measurement is then a projection onto the positive eigenspace of $\Delta$. Importantly, both the optimal strategy and the achievable performance depend only on the operator $\Delta$. 

Distinguishability is fundamentally a property of the measurements available as much as it is a property of the states themselves. In realistic quantum technologies, measurements are often constrained by the underlying hardware or by fault-tolerance requirements. Understanding how such physical restrictions limit distinguishability is therefore of practical importance. Resource theories~\cite{RevModPhys.91.025001, Gour_2025} provide a natural framework for addressing such questions: they identify operations that are freely implementable (allowed) under some physical restrictions, and quantify the operational advantages enabled by objects beyond these constraints. Among these, the resource theory of non-stabilizerness (or “magic”)~\cite{PhysRevA.71.022316, Veitch_2014, PhysRevLett.118.090501} has been shown to play a central role in fault-tolerant quantum computation~\cite{Gottesman1998, AaronsonGottesman2004, gottesman1997stabilizer, Campbell2017}, among other applications \cite{PhysRevLett.132.210602, Cao_2025, zamora2025semi, hwfq-ytkb, loio2026quantumstatedesignsmagic}. Its allowed operations, known as stabilizer circuits, consist of Clifford unitaries, computational-basis measurements, and stabilizer-state preparation. While these operations admit efficient classical simulation~\cite{Gottesman1998, AaronsonGottesman2004}, they are insufficient for universal quantum computation, which requires the consumption of non-stabilizer resources. Although the computational role of non-stabilizerness~\cite{PhysRevA.71.022316, Beverland_2020, Campbell2017, PRXQuantum.2.010345} and its characterization~\cite{Veitch_2014, PhysRevLett.118.090501, PhysRevLett.128.050402, Sonya_Tarabunga_2025, 2s3j-t22p, varela2026predictingmagicmeasurements, PRXQuantum.3.020333} have been extensively studied, its operational role in restricting quantum measurements remains largely unexplored \cite{kwon2025nonstabilizernessmagicclassicallysimulatable}.

In this Letter, we investigate minimum-error quantum state discrimination when the player is restricted to only stabilizer circuits. We show that, for fixed stabilizer circuits, access to arbitrary stabilizer ancillas cannot increase the minimum-error success probability. In contrast, this is not true for adaptive stabilizer circuits. For qubit systems, we then derive analytical expressions for the minimum-error  success probability with arbitrary qubit ancillas for both classes of circuits, showing that ancillary non-stabilizerness can ---in some cases--- improve discrimination with adaptive circuits, while again never providing an advantage for fixed circuits.

We further analyze discrimination when non-stabilizerness is incorporated directly into the measurement operators, showing that the success probability admits a semidefinite programming formulation. Analytical bounds for qubits that continuously interpolate between stabilizer restrictions and Helstrom are then derived. Finally, we illustrate the operational implications of our findings in quantum random access codes and in bounds on the achievable fidelity of unitary synthesis with a finite number of magic states.

\section{Framework} 
Let $\mathcal{H}_{n}=(\mathbb{C}^{2})^{\otimes n}$ be the $n$-qubit Hilbert space of dimension $d=2^n$, and $\mathcal{D}(\mathcal{H}_n)$ the set of density operators on $\mathcal{H}_n$. Then, let $\mathcal{P}_{n}$ be the set of $n$-qubit Pauli strings --- these are all the $n$-fold tensor product of the qubit Pauli operators $\mathbb{I}, X, Y, Z$, where $\vert \mathcal{P}_n \vert = 4^{n}$. The set of unitary operators that map Pauli strings to Pauli strings are the so-called Clifford unitaries, $\mathcal{C}_n$. The $n$-qubit pure stabilizer states are then those states that can be generated from $ \ket{0}^{\otimes n}$ and $\mathcal{C}_n$. Equivalently, a state $\ket{\psi}$ is a stabilizer state if there exists a subgroup, \hbox{$S \subset \mathcal{P}_{n}$}, of $d$ mutually commuting Pauli-strings such that \hbox{$\bra{\psi} P \ket{\psi} = \pm 1 \forall~ P \in S$} and \hbox{$\bra{\psi} P \ket{\psi} = 0 ~\forall~P \in \mathcal{P}_n \setminus S$}. The free states of the resource theory of non-stabilizerness are the convex hull of the pure stabilizer states, $\hbox{${\rm STAB}(n) \coloneq {\rm conv}\big\{ C \ket{0}^{\otimes n} : C \in \mathcal{C}_n \big\}$}$, the allowed operations are stabilizer circuits. Where necessary, we will denote the stabilizer set over any $n$ as ${\rm STAB}$. All measurements are modelled by positive operator valued measure (POVMs) \cite{nielsen_chuang_2010}. 

On an $n$ qubit state, without appending a stabilizer state, POVMs of the form 
\begin{equation*}
    \big\{ \Pi^U_a = U^\dagger \ketbra{a} U \big\}_{a}~,
\end{equation*}
can be measured by a stabilizer circuit (by applying $U$ then measuring in the computational basis), where $U \in \mathcal{C}_n$ and $a \in \{0,1\}^n$ is an n-bit-string.
On an $n$-qubit state (system $A$) with a $k$ qubit state $\omega$ appended (system $B$), POVMs of the form 
\begin{equation*}
    \big\{ M^{U, \omega}_a = \textrm{tr}_B \big[ U^\dagger \ketbra{a} U (\mathbb{I}_A \otimes \omega_B) \big] \big\}_{a}~,
\end{equation*}
can be measured, where now $U \in \mathcal{C}_{n+k}$ and $a \in \{0,1\}^{n+k}$ is an ($n + k$) bit-string. Under stabilizer circuits alone, it must be that $\omega \in {\rm STAB}(k)$. However, the case that $\omega \notin {\rm STAB}$, such that non-stabilizer circuits can be simulated, will later be investigated.  
  
The measurements described above are \emph{fixed}, i.e., they are fully specified before the circuit begins. We also consider \emph{adaptive stabilizer circuits}, where intermediate measurement outcomes determine the subsequent operations. We here specifically consider a measurement on an ancillary system selecting which measurement is subsequently performed on the primary system. For an $n$-qubit state with a $k$-qubit ancilla $\omega$, this corresponds to POVMs of the form
\begin{equation*}
    \big\{ T^{U, \omega, \{V\}}_{a,b} = \textrm{tr}_B\big[ U^\dagger \big( V_b^\dagger \ketbra{a} V_b \otimes \ketbra{b} \big) U (\mathbb{I}_A \otimes \omega_B) \big] \big\}_{a,b}~,
\end{equation*}
where $a \in \{0,1\}^n, b \in \{0,1\}^k$ are $n$ and $k$ bit-strings output from computational basis measurements on the primary and ancillary system respectively, $U \in \mathcal{C}_{n+k}$, and $V_b \in \mathcal{C}_n~\forall~b \in \{0,1\}^k$. See SM.~\ref{SM: Measurements in the Resource Theory of Non-stabilizerness} for a more detailed overview of measurements from fixed and adaptive stabilizer circuits.

\begin{figure}[t!]
    \centering
    \includegraphics[width=0.9\linewidth]{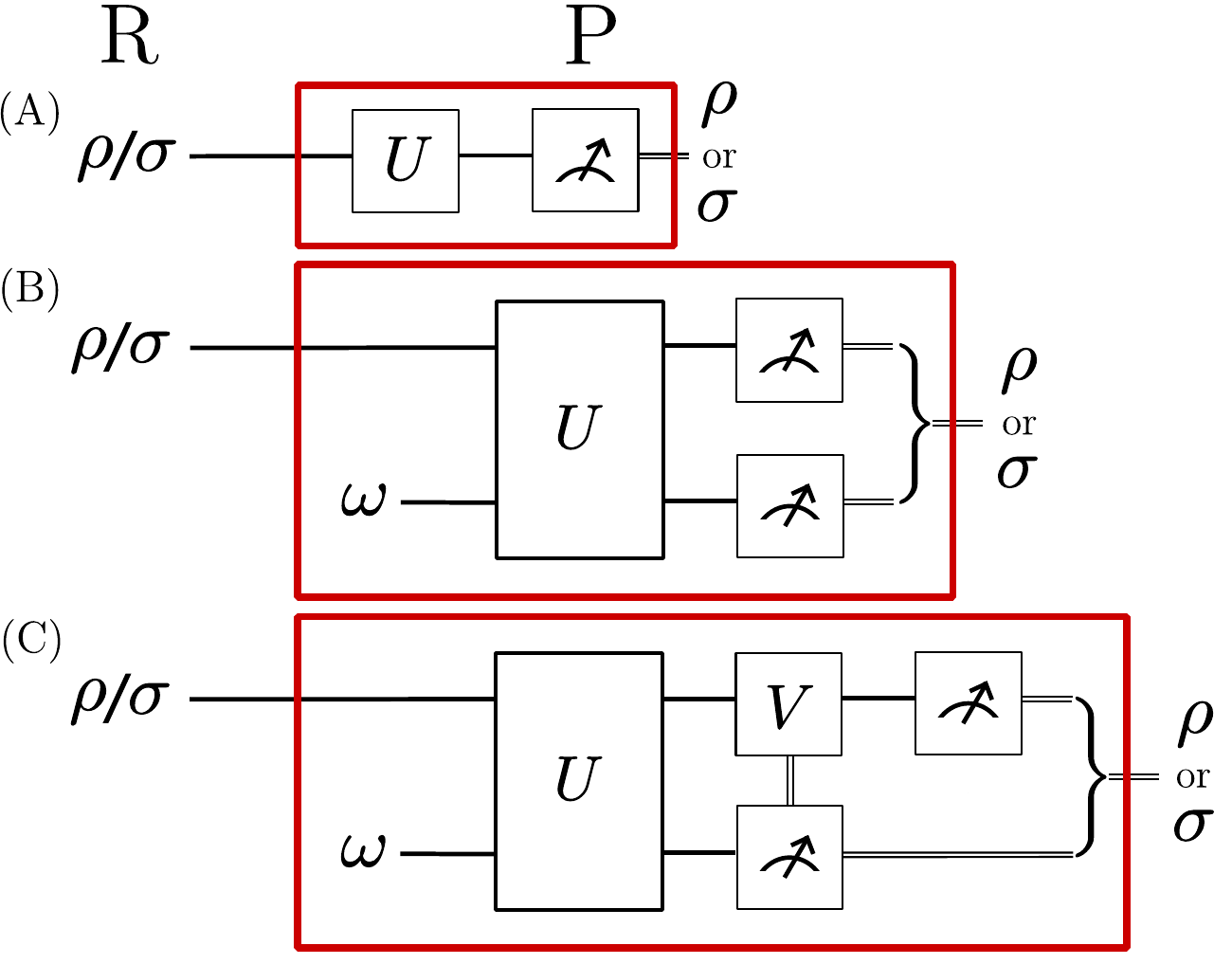}
    \caption{A graphical depiction of the task of QSD, where the referee (R) gives the player (P) either the state $\rho$ or the state $\sigma$, under (A) fixed stabilizer circuits with no ancilla ($P^{\rm stab}_{\rm suc} (\Delta)$), (B) fixed stabilizer circuits with an ancilla $(P^{\rm stab,f}_{\rm suc}(\Delta, \omega))$, and (C) adaptive stabilizer circuits $(P^{\rm stab,a}_{\rm suc}(\Delta, \omega))$,}\label{figure: schematic of the task}
\end{figure}

\section{QSD Under Stabilizer Circuits}
We now evaluate the minimum-error success probability for QSD when the player is restricted to fixed or adaptive stabilizer circuits, while allowing the referee to prepare arbitrary quantum states. As permitted by stabilizer circuits, the player is allowed to append an arbitrary stabilizer state to the unknown input. While in the unrestricted setting ancillary states cannot improve the minimum-error success probability, this need not hold under measurement restrictions. We therefore first investigate the success probability as a function of the appended ancillary state. 

Assume that $\Delta$ is the difference of two $n$-qubit density operators, with each state given to the player with equal probability, and let $\omega \in \mathcal{D}(\mathcal{H}_k)$ for some $k$. Then the following figures of merit are considered 
\begin{align*}
    P^{\rm stab}_{\rm suc} (\Delta) &= \frac{1}{2} + \max_{U \in \mathcal{C}_n} \frac{1}{4} \sum_{a \in \{0,1\}^{n}} \big\vert \textrm{tr} \big[ \Pi_a^U (\Delta) \big] \big\vert, \\
    P^{\rm stab, f}_{\rm suc} (\Delta, \omega) &= \frac{1}{2} + \max_{U \in \mathcal{C}_{n+k}} \frac{1}{4} \sum_{a \in \{0,1\}^{n+k}} \big\vert \textrm{tr} \big[ M^{U,\omega}_a (\Delta) \big] \big\vert, \\
    P^{\rm stab, a}_{\rm suc} (\Delta, \omega) &= \frac{1}{2} + \max_{ \substack{U \in \mathcal{C}_{n+k} \\ \{V\} \in \mathcal{C}_n}} \frac{1}{4} \sum_{\substack{b \in \{0,1\}^k \\ a \in \{0,1\}^n}} \big\vert \textrm{tr} \big[ T^{U,\omega, \{V\}}_{a,b} (\Delta) \big] \big\vert,
\end{align*}
which are the minimum error success probability under fixed stabilizer circuits without an ancillary system, fixed stabilizer circuits with $\omega$ in the ancillary system, and adaptive stabilizer circuits with $\omega$ in the ancillary system, respectively. 
We note that, due to each figure of merit being convex, each will be maximised when $\omega$ is a pure-state. See Fig.~\ref{figure: schematic of the task} for a schematic of the tasks, and see Appendix~\ref{appendix main text: QSD With Coarse Grained Measurement} for details of how these figures of merit include an optimisation over all possible assignments of POVM elements to states.

Now, as a first result, we show in SM.~\ref{appendix: lemma: STABs don't increase sucess prob} that for all $\Delta$ a stabilizer state in the ancilla can never provide an advantage in QSD under fixed stabilizer circuits. 
\begin{theorem}\label{theorem: STABs don't increase sucess prob}
    $P_{\rm suc}^{\rm stab}(\Delta) = P_{\rm suc}^{\rm stab, f}(\Delta, \omega) ~\forall~\omega \in {\rm STAB}$. 
\end{theorem}
In appendix~\ref{appendix: counter example} we give a counter example to prove that this does not in general extend to adaptive circuits. Hence, adaptivity alone can be seen to provide an advantage in QSD when restricted to perform only stabilizer circuits.

\subsection{Qubit Discrimination}
In this case, we consider \hbox{$\Vec{\Delta} \coloneq n_\rho - n_\sigma$}, which is the difference of the Bloch vectors of qubits $\rho$ and $\sigma$. 
Under a restriction to only stabilizer circuits without an ancilla, the achievable minimum error success probability is shown in SM.~\ref{appendix: proof of qubit fixed stab measurement no ancilla} to be given by the following expression. 
\begin{lemma}\label{lemma: qubit fixed stab measurement no ancilla}
        $P_{\rm suc}^{\rm stab}(\Vec{\Delta}) = \frac{1}{2} + \frac{1}{4} \vert \vert \Vec{\Delta} \vert \vert_\infty,$
    where $\vert \vert \Vec{\Delta} \vert \vert_\infty$ is the vector infinity norm i.e., the largest component of the absolute value of $\Vec{\Delta}$. 
\end{lemma}
Lemma~\ref{lemma: qubit fixed stab measurement no ancilla} makes intuitive sense: the player chooses the direction of $\Vec{\Delta}$ that is largest and measures in that direction (qubit stabilizer measurements are Pauli $X$, $Y$ or $Z$ measurements). Moreover, it can be seen that $P_{\rm suc}^{\rm stab}(\Vec{\Delta})=1$ if and only if $\Vec{\Delta}$ is the difference of two pure, orthogonal stabilizer states. 

We now introduce an arbitrary ancillary qubit, which could include a finite amount of non-stabilizerness, and quantify its effect on the minimum-error success probability. Interestingly, as shown in SM.~\ref{appendix: result: for qubits, not about of resource is useful}, no amount of non-stabilizerness improves the success probability under fixed measurements for any $\Vec{\Delta}$.
\begin{theorem}\label{result: for qubits, not about of resource is useful}
$P_{\rm suc}^{\rm stab}(\Vec{\Delta}) = P_{\rm suc}^{\rm stab, f}(\Vec{\Delta}, \omega)~\forall~\omega \in \mathcal{D}(\mathcal{H}_1)$.
\end{theorem}
Hence, for a single ancillary qubit, one must consider adaptive measurements to gain any advantage. The following result, proved in SM.~\ref{appendix: appending any qubit state with feedfoward}, captures the maximum success probability achievable given access to an arbitrary qubit ancilla and adaptive measurements. 
\newpage
\begin{theorem}\label{result: appending any qubit state with feedfoward}
    Let $n_\omega$ be the Bloch vector of $\omega \in \mathcal{D}(\mathcal{H}_1)$ and $\vert \vert \Vec{v} \vert \vert_{2-{\rm ky}}$ the $2$nd vector KY-Fan norm i.e., the sum of largest two components of the absolute value of $\Vec{v}$, then 
    \begin{equation*}
        P_{\rm suc}^{\rm stab,~ a}(\Vec{\Delta}, \omega) =  \frac{1}{2} + \frac{1}{4} \max \big\{\vert \vert \Vec{\Delta} \vert \vert_\infty, \vert \vert \Vec{\Delta} \vert \vert_{2-\rm{ky}} \vert \vert \Vec{n}_\omega \vert \vert_{2-\rm{ky}}{/2} \big\}
    \end{equation*}
\end{theorem}

\begin{figure}[t!]
    \centering
    \includegraphics[width=0.85\linewidth]{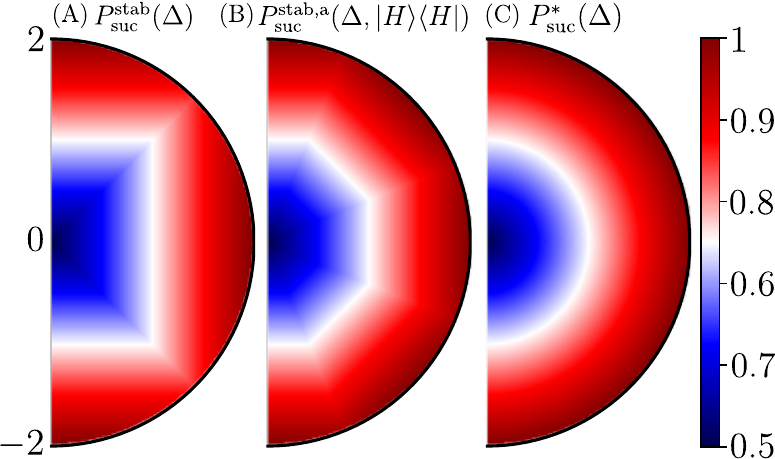}
    \caption{The minimum error success probability for all $\Vec{\Delta}$ in the $xz$-plane for (A) stabilizer circuits, (B) stabilizer circuits with a single $H$-type magic state and adaptive measurements, and (C) unrestricted measurements (Helstrom's).   }\label{figure: semi-circle figure} 
\end{figure}

From Result~\ref{result: appending any qubit state with feedfoward} it can be seen that a $H$-type magic state, $\ketbra{H} = \frac{1}{2}\big(\mathbb{I} + \frac{1}{\sqrt{2}}(X + Z)\big)$ for example \cite{PhysRevA.71.022316}, is always the optimal ancillary state, as such states satisfy $\vert \vert  \Vec{n}_\omega \vert \vert_{2-{\rm ky}}={\sqrt{2}}$ which is the maximum possible value. 
Operationally, this can be seen as using a single $T$-gate in setting the measurement basis on the unknown system, implemented via a $H$-type magic state and magic state injection. However, even with $\omega=\ketbra{H}$, Result~\ref{result: appending any qubit state with feedfoward} shows that there exists $\vec \Delta$ for which adaptive circuits do not lead to an advantage. For such $\vec \Delta$, setting the optimal measurement basis therefore needs only Clifford unitaries. See SM.~\ref{appendix: When does Non-stabilizerness provide an advantage} for an analysis of for what $\vec \Delta$ non-stabilizerness provides an advantage.

In Fig.~\ref{figure: semi-circle figure} a visual comparison of Lemma~\ref{lemma: qubit fixed stab measurement no ancilla}, Result~\ref{result: appending any qubit state with feedfoward} and Helstrom's bound is given for $\Vec{\Delta}$ confined to the $xz$-plane for positive $x$ $\big($using $P^{\rm stab}_{\rm suc}(\vec \Delta) = P^{\rm stab}_{\rm suc}(-\vec \Delta)\big)$ with $\omega=\ketbra{H}$ (or any other $H$-type magic state). Each point in the semi-circles corresponds to a valid $\Vec{\Delta}$, with the colour determining the achievable minimum error success probability under the respective measurement restrictions. It can be seen that with access to a single $H$-type magic state the Helstrom's bound begins to be recovered. In appendix~\ref{appendix: Recovering Helstrom's Bound} we assess for what orthogonal pure states Helstrom's bound is recovered given access to a single $H$-type magic state. 

Finally, we note that as $\Vert \vec n \Vert_{2-{\rm ky}} \leq 1$ for all qubits in ${\rm STAB(1)}$, Result~\ref{theorem: STABs don't increase sucess prob} can here be extended to include adaptive stabilizer circuits.
\begin{corollary}\label{corollary: STABs don't increase sucess prob}
    $P_{\rm suc}^{\rm stab}(\vec \Delta) = P_{\rm suc}^{\rm stab, f}(\vec \Delta, \omega) = P_{\rm suc}^{\rm stab, a}(\vec \Delta, \omega)$ \newline$~ ~\forall~\omega \in {\rm STAB}(1)$. 
\end{corollary}

\section{QSD with Non-stabilizer Measurements}
We now consider {introducing non-stabilizerness directly into the POVM elements, as opposed to introducing it via an ancillary system. For a given two outcome discriminating POVM, we quantify its non-stabilizerness by normalising the POVM elements and then using the trace distance \hbox{$\mathcal{M}(\rho) = \min_{\sigma \in \mathrm{STAB}} \frac{1}{2} \| \rho - \sigma \|_1$}~\cite{Cao_2025, palhares_2026}.
See SM~\ref{sm: SDP for non-stabilizer measurements} for more details. }

{When using a given POVM $\{M_0^\mu,M_1^\mu\}$, with non-stabilizerness $\mathcal{M}(M_i^\mu)$ restricted to $\mu \in [0,\mu_{max}]$, the optimal minimum error success probability is captured by:}
\begin{align}
&P_{\text{suc}}^\mu(\Delta) = \max_{M^\mu_0, M^\mu_1} \frac{1}{2} + \frac{1}{2} \text{Tr}(\Delta M^\mu_0),\\
&\text{s.t.}\quad\mathcal{M}(M^\mu_i)\leq \mu \quad \text{for} \quad i\in\{0,1\}\label{eq:opt_magic_constraits}.
\end{align}
In SM \ref{sm: SDP for non-stabilizer measurements}, we show that finding $P_{\text{suc}}^\mu(\Delta)$ can be cast as a semidefinite program, where the overall computational complexity is governed by the size of the stabilizer polytope description used to encode the constraints~\eqref{eq:opt_magic_constraits}. 

In the specific case of qubits, we derive an analytical lower bound on $P_{\text{suc}}^\mu(\Delta)$. To characterize the solution, it is necessary to partition the optimization: let $\vec s$ denote the Bloch vector of the closest stabilizer state to the optimal measurement, and define the subspace $\mathcal{A}$ as the subspace spanned by the coordinates for which $s_i>0$, where $s_i$ is the $i$th element of $\vec s$. Then, $\text{dim}(\mathcal{A})\in\{1,2,3\}$ will indicate whether $\vec{s}$ lives in a vertex, edge or face of the stabilizer polytope. Under the no-leakage ansatz, it is then assumed that the optimal measurement vector belongs to this same subspace $\mathcal{A}$. The optimisation can then be solved separately within each  possible subspace $\mathcal{A}$, with a bound on $P_{\text{suc}}^\mu(\Delta)$ obtained by maximizing over all the feasible subspaces $\mathcal{A}$. This leads to the following result, which is proved in SM.~\ref{appendix: QSD for qubits with magic measurements}.
\begin{theorem}\label{result: QSD for qubits with non-stabilizer measurements}For measurements $\{M_0^\mu,M_1^\mu\}$ with $\mathcal{M}(M_i^\mu) \leq \mu~\forall~i$, the success probability is tightly bounded below by:
\begin{equation}
P_{\mathrm{suc}}^\mu \geq \max_{\mathcal A}P^\mathcal{A}_{\mathrm{suc}}(\mu),
\end{equation}
where {$P^\mathcal{A}_{\mathrm{suc}}(\mu)$ denotes the optimal success probability when assuming that the closest stabilizer state has active support on the subspace $\mathcal{A}$}:
\begin{equation}\label{eq: Psucc_Nostab_m_qbits_mt}
P^\mathcal{A}_{\mathrm{suc}}(\mu) = 
\begin{cases}
\frac{1}{2} + \frac{1}{4N}(S+C_N\vert \vert\vec\Delta\vert \vert_{\mathcal{A},1}), & \text{if } C_N < \frac{\vert \vert\vec\Delta\vert \vert_{\mathcal{A},1}}{\vert \vert\vec\Delta\vert \vert_{\mathcal{A},2}} \\
\frac{1}{2} + \frac{1}{4}\vert \vert\vec{\Delta}\vert \vert_{\mathcal{A},2}, & \text{if } C_N \ge \frac{\vert \vert\vec\Delta\vert \vert_{\mathcal{A},1}}{\vert \vert\vec\Delta\vert \vert_{\mathcal{A},2}}
\end{cases},
\end{equation}
where $S=\sqrt{(N\vert \vert\vec\Delta\vert \vert_{\mathcal{A},2}^2 - \vert \vert\vec\Delta\vert \vert_{\mathcal{A},1}^2)(N-C_N^2)}$; ~\hbox{$\vert \vert\cdot\vert \vert_{\mathcal{A},p}$} denotes the vector $p$-norm restricted to $\mathcal{A}$; $N = \text{dim}(\mathcal A)$; and $C_N = 1 +2\sqrt{N}\mu$ defines the surface with non-stabilizerness $\mu$ restricted to $\mathcal{A}$. 
\end{theorem}
In Figure~\ref{figure:Psucc_mu} we compare the analytical expression of Result~\ref{result: QSD for qubits with non-stabilizer measurements} with the SDP (detailed in SM \ref{sm: SDP for non-stabilizer measurements}) for antipodal states in representative geodesics on the Bloch sphere. Namely, we consider the success probability for various $\mu$  when performing discrimination between states ranging from $\ket{T}$ to $\ket{0}$ ($\ket{T}$ to $\ket{H}$) on a Bloch sphere geodesic with their respective orthogonal complement (see Appendix~\ref{appendix: Geodesics} for details). Along the symmetry axes of the stabilizer octahedron (for example, $\ket{T}\rightarrow \ket{0}$), the analytical solution exactly reproduces the SDP optimum, confirming that the no-leakage assumption is tight. In contrast, for the $\ket{T}\rightarrow \ket{H}$ geodesic, the SDP exploits leakage into inactive dimensions, so Result~\ref{result: QSD for qubits with non-stabilizer measurements} becomes a strict lower bound.

\begin{figure}[t!]
    \centering
    \includegraphics[width=1\columnwidth]{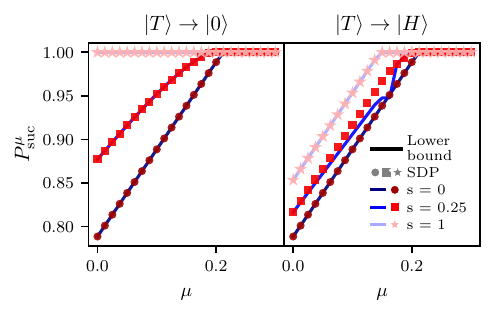}
    \caption{\textbf{Optimal success probability $P_{\mathrm{suc}}^\mu$ versus measurement non-stabilizerness $\mu$.} Discrimination of antipodal states along trajectories $\ket{T} \to \ket{0}$ (left, symmetric) and $\ket{T} \to \ket{H}$ (right, asymmetric) for $s \in \{0, 0.25, 1\}$ on the Bloch sphere (see Appendix~\ref{appendix: Geodesics}). The analytical bound (lines) saturates the exact SDP maximum (markers) along the stabilizer octahedron's symmetry axes (left). Conversely, for asymmetric trajectories (right), the SDP exploits measurement leakage to strictly exceed the analytical bound at intermediate $\mu$. } \label{figure:Psucc_mu}  
\end{figure}
When $\mu$ is sufficiently large --- meaning the measurement have high non-stabilizerness --- the measurement vector can align with $\vec{\Delta}$. Hence, ($C_N \ge \vert\vert\vec{\Delta}\vert\vert_{\mathcal{A},1}/\vert\vert\vec{\Delta}\vert\vert_{\mathcal{A},2}$) and Helstrom bound is recovered. Conversely, when $\mu = 0$, the measurement is confined to the stabilizer polytope. Here, $C_N = 1$ and the active subspace collapses to a single vertex ($N=1$). In this regime, the formula evaluates to the largest component of $\Vec{\Delta}$, which exactly recovers lemma~\ref{lemma: qubit fixed stab measurement no ancilla}.

\section{Applications: Random Access Codes}
As a first example of an application, we consider $n\to1$ quantum random access codes~\cite{ambainis2009}. In this task, Alice encodes a bit-string $\mathbf{x} = (x_1,\dots,x_n) \in \{0,1\}^n$ into some state $\rho_\mathbf{x}$ and sends it to Bob. Bob then aims to determine the $jth$ bit of $\mathbf{x}$ by measuring $\rho_\mathbf{x}$. To achieve this, Bob measures the $jth$ two outcome POVM from some set $\{~\{ M_{b \vert j} \}_{b \in \{0,1\}}\}_{j \in [n]}$.
 
Bob's average success probability is then \hbox{$P_{g}= \frac{1}{2^nn}\sum_{j\in[n]}\sum_{x\in\{0,1\}^n}\text{Tr}(\rho_xM_{x_j|j})$}. In Ref~\cite{zamora2025semi}, bounds on $P_g$ were obtained when the states were restricted to be only stabilizer states. Here, we solve the complementary problem of obtaining bounds when the measurements are restricted to be stabilizer. See SM.~\ref{appendix: QRAC bound} for the proof.
\begin{theorem}\label{result: QRAC bound}
     Let us split Alice's bistrings in two disjoint sets: $\mathbb X_0^{(j)} =  \{{x\in\{0,1\}^n} |x_j =0\}$ and $\mathbb X^{(j)}_1 =  \{{x\in\{0,1\}^n} |x_j =1\}$. When restricted to only stabilizer measurements, the maximum value of $P_g$, denoted $P_{g}^{\mathrm{STAB}}$, is given by 
     \begin{equation}
         P_{g}^{\mathrm{STAB}} = \frac{1}{2} +\frac{1}{4n}\sum_{j\in [n]}\vert\vert\vec{\Delta}^{(j)}\vert\vert_\infty.
     \end{equation}
     where $\Delta^{(j)} = \sigma^{(j)}_0- \sigma^{(j)}_1=\frac{1}{2}\vec{\Delta}^{(j)}\cdot\vec{\sigma}$ with $\vec{\sigma}$ a vector of Pauli matrices and $\sigma^{(j)}_{i} = \frac{1}{2^{n-1}}\sum_{x\in\mathbb{X}^{(j)}_i}\rho_x$ effective states.
\end{theorem}
Since $\vert\vert \vec{\Delta}^{(j)}\vert\vert_\infty\leq\vert\vert \vec{\Delta}^{(j)} \vert\vert_2$, restricting Bob to stabilizer measurements is generally suboptimal. Together with Ref.~\cite{zamora2025semi}, Result~\ref{result: QRAC bound} shows that optimal QRACs require non-stabilizerness in both the preparation and measurement stages, identifying measurement non-stabilizerness as an operational resource for their implementation.

\section{Applications: Simulation Fidelity}
Any unitary can be synthesized using Clifford and $T$ gates. Since stabilizer circuits comprise only Clifford operations, implementing a $T$ gate requires magic-state injection, where each $T$ gate consumes a single $H$-type magic state $\ketbra{H}$. Consequently, a finite supply of magic states limits the set of implementable unitaries, as it limits the number of implementable $T$ gates. This motivates the question of how accurately an arbitrary unitary can be synthesized within the resource theory of non-stabilizerness given access to only a finite number of $H$-type magic states \cite{Beverland_2020}. Using Result~\ref{result: appending any qubit state with feedfoward}, we derive an upper bound on the fidelity $F(U,V)=1/d^2\,|\mathrm{tr}(U^\dagger V)|^2$
with which an arbitrary qubit unitary $V$ can be synthesized by a unitary $U$ given access to a single $H$-type magic state (equivalently, a single $T$ gate).
\begin{theorem}\label{result: simulation fidelity}
     Let $V,U$ be qubit unitaries where \hbox{$U = C_1 T C_2 : C_1, C_2 \in \mathcal{C}_1$}. Then 
     \begin{equation}
         F(U,V) \leq P^{\rm stab, a}_{\rm suc}(\Delta_V, \ketbra{H}) 
     \end{equation}
     where $\Delta_V = V^\dagger(\ketbra{0}-\ketbra{1})V$. 
\end{theorem}
See SM.~\ref{appendix: result simulation fidelity} for the proof. Result~\ref{result: simulation fidelity} bounds the fidelity with which a unitary $V$ can be simulated by a unitary $U$ --- which contains a single $T$ gate --- by the maximum probability with which the state $V^\dagger \ketbra{0} V$ and its orthogonal complement can be discriminated under adaptive stabilizer circuits and access to a single H-type magic state. 

\section{Discussion}
We have shown that quantum state discrimination provides an operational framework for quantifying the measurement power enabled by non-stabilizerness. Restricting measurements to stabilizer circuits defines a weaker notion of distinguishability than the Helstrom limit, with non-stabilizer resources then progressively restoring the achievable success probability. We showed that adaptive strategies are necessary to exploit ancillary non-stabilizerness, establishing an operational separation between fixed and adaptive stabilizer circuits. Moreover, our SDP formulation provides a general framework for studying discrimination under bounded measurement non-stabilizerness in arbitrary dimensions. Taken together, our results connect the resource theory of non-stabilizerness with one of the most fundamental operational tasks in quantum information, establishing quantum state discrimination as a natural benchmark for characterizing the power of restricted measurements.

\begin{acknowledgments}
 We thank A. de Oliveira Junior for helpful discussions and feedback. We acknowledge financial support from the Simons Foundation (Grant No. 1023171, R.C.), the Brazilian National Council for Scientific and Technological Development (CNPq, Grants No. 403181/2024-0 and 301687/2025-0), the National Institute of Science and Technology for Applied Quantum Computing through CNPq process No. 408884/2024-0, the Financiadora de Estudos e Projetos (Grant No. 1699/24 IIF-FINEP), and the Coordenação de Aperfeiçoamento de Pessoal de Nível Superior -- Brasil (CAPES) -- Finance Code 001.\textbf{ Code availability.}  The code used to generate the results in this paper is available on GitLab~\cite{zamora2026code}.
\end{acknowledgments}

\appendix

\section{QSD With Coarse Grained Measurement}\label{appendix main text: QSD With Coarse Grained Measurement}
In the minimum error form of QSD with two states, a two outcome POVM is used to model the measurement. If the player gets one outcome they guess $\rho$, if they get the other they guess $\sigma$. If, instead, a POVM with $N$ outcomes is measured it must be coarse grained into a two outcome measurement. This can occur before the measurement, or after (via classical post processing). However, there exists exponentially many ($2^{N}$) different options for coarse grainings the POVM into a two outcome measurement. To handle this, the following Lemma, initially noted in \cite{Matthews_2009}, can be used to calculate the success probability for performing QSD under the optimal coarse graining. 
\begin{lemma}\label{lemma: general state discrimination}
    If using a POVM $\{M_i \}_{i=1}^N$ for perform QSD with $\Delta$, the maximum success probability under all possible coarse grainings is
    \begin{equation}
        P_{\rm suc}(\Delta) = \frac{1}{2} + \frac{1}{4} \sum_{i=1}^N \vert \textnormal{tr}\big[ M_i \Delta \big] \vert.
    \end{equation}
\end{lemma}
See Supplementary Material (SM)~\ref{sm: proof of general state discrimination} for a proof.

\section{Recovering Helstrom's Bound}\label{appendix: Recovering Helstrom's Bound}
If $\Delta$ is the difference of two orthogonal pure states, $P^*_{\rm suc}(\Delta)=1$. When restricted to only stabilizer circuits on qubits, this continues to hold if and only if $\Delta$ is the difference of two orthogonal pure stabilizer states. 

In general, the maximum achievable minimum error success probability is reduced when restricting to stabilizer circuits. However, in Result~\ref{result: appending any qubit state with feedfoward} it was shown that, for qubits, this restriction can be partially overcome given access to adaptive stabilizer circuits and a source of non-stabilizerness. Although, even with the optimal source of non-stabilizerness i.e., a $H$ type magic state, Helstrom's bound is not in general recovered. To see exactly when Helstrom's bound can be recovered for orthogonal pure states, one must consider when $P_{\rm suc}^{\rm STAB, a}(\Delta, \omega)=1$ if $\omega=\ketbra{H}$. This occurs if $\vert \vert \Vec{\Delta} \vert \vert_\infty=2$ or $\vert \vert \Vec{\Delta} \vert \vert_{2-{\rm ky}}=4/\sqrt{2}$. The former case corresponds to when $\Vec{\Delta}$ is the difference of two orthogonal pure stabilizer states; the latter occurs when $\vert \Vec{\Delta}^\downarrow \vert = (2/\sqrt{2}, 2/\sqrt{2}, 0)^t$, meaning that $\Vec{\Delta}$ is the difference of two orthogonal $H$-type magic states. Hence, these are the only orthogonal pure states for which Helstrom's bound is recovered given access to a single $H$-type magic state. 

\section{QSD under Adaptive Stabilizer Circuits}\label{appendix: counter example}
Here, we show an example of a QSD tasks with adaptive stabilizer circuits in which a stabilizer state in an ancilla allows the minimum error success probability to be improved as compared to fixed stabilizer circuits. 

Consider the following two two-qubit states on a space $\mathcal{H}^{A_1}_1 \otimes \mathcal{H}^{A_2}_1$:
\begin{equation}
    \begin{split}
        \rho &= \frac{1}{2}(\ketbra{00}_{A_1A_2} + \ketbra{1+}_{A_1A_2}) \\
        \sigma &= \frac{1}{2}(\ketbra{01}_{A_1A_2} + \ketbra{1-}_{A_1A_2}).
    \end{split}
\end{equation}
Now, apply the following adaptive stabilizer circuit: append $\ketbra{0} \in {\rm STAB}(1)$ in $\mathcal{H}^1_B$, then perform $C_{A_1}X_B$ (controlled-X between $A_1$ and $B$) before measuring $B$ in the computational basis. If the outcome is $0$, apply $\mathbb{I}_{A_1A_2}$; if the outcome is $1$, apply $\mathbb{I}_{A_1}H_{A_2}$. Finally, measure $A_1$ and $A_2$ in the computational basis. 

Such a circuit will lead to an effective POVM of the form
\begin{equation}
    \big\{ \ketbra{00}, \ketbra{01}, \ketbra{1+}, \ketbra{1-} \big\},
\end{equation}
which is able to perfectly distinguish between $\rho$ and $\sigma$ i.e., $P^{\rm stab, a}_{\rm suc}(\Delta, \ketbra{0})=1$. However, when measuring under fixed stabilizer circuits it can be shown that $P^{\rm stab}_{\rm suc}(\Delta)=3/4$, such that there exists an example where $P^{\rm stab}_{\rm suc}(\Delta) < P^{\rm stab, a}_{\rm suc}(\Delta, \omega)$ when $\omega \in {\rm STAB}$. Finally, we note that this discrimination tasks could equivalently be achieved with probability $1$ without access to $\ket{0}$, by performing adaptive measurements directly on the primary system. However, we do not consider such a framework in this work, instead requiring all adaptive to be performed on an ancillary system. See SM.~\ref{SM: counter example} for more detail.  

\section{Generation of Geodesics Trajectories for Fig.~\ref{figure:Psucc_mu} }\label{appendix: Geodesics}
Parametrizing the geodesic on the Bloch sphere defined by the vectors $\vec{n}_{start}$ and $\vec{n}_{end}$ separated by an angle $\theta = \arccos(\vec{n}_{\text{start}} \cdot \vec{n}_{\text{end}})$ via spherical linear interpolation~\cite{Shoemake_1985}: $\vec{n}(s) = \frac{\sin((1-s)\theta)}{\sin(\theta)} \vec{n}_{\text{start}} + \frac{\sin(s\theta)}{\sin(\theta)} \vec{n}_{\text{end}}$. In Fig.~\ref{figure:Psucc_mu} we compare the analytical bound from Result~\ref{result: QSD for qubits with non-stabilizer measurements} with the SDP optimum detailed in SM~\ref{sm: SDP for non-stabilizer measurements} for a QSD task considering antipodal states (defined by $\vec{n}(s)$ and $-\vec{n}(s)$, for $s\in\{0,0.25,1\}$). This is done for the geodesics joining $\ket{T}$ ($\vec{n}_{\text{start}} = \frac{1}{\sqrt{3}}(1,1,1)$) with $\ket{0}$ ($\vec{n}_{\text{end}} = (0,0,1)$) as a function of the measurement non-stabilizerness $\mu$. Note that due to the symmetry of these states,  $P_{\mathrm{suc}}^\mu = \max_{\mathcal A}P^\mathcal{A}_{\mathrm{suc}}(\mu)$. Therefore it also constitutes the solution obtained through the  SDP proposed in SM~\ref{sm: SDP for non-stabilizer measurements}. Finally, we consider the geodesic between $\ket{T}$ and $\ket{H}$ ($\vec{n}_{\text{end}}= \frac{1}{\sqrt{2}}(1,0,1)$) showing that Eq.~\eqref{eq: Psucc_Nostab_m_qbits_mt} is in general a strict lower bound of the SDP solution.

\bibliographystyle{apsrev4-1}
\bibliography{mainTextBib}

\clearpage


\onecolumngrid
\begin{center}
    \Large \bfseries {Supplementary Material: State Discrimination With Stabilizer Measurements}
\end{center}
\vspace{1em}

\let\addcontentsline\oldaddcontentsline %

\begingroup
\parskip=0pt
\setcounter{tocdepth}{3}
\tableofcontents
\endgroup

\onecolumngrid

\section{Measurements Under Stabilizer Circuits}\label{SM: Measurements in the Resource Theory of Non-stabilizerness}
Measurements are modelled by Positive Operator Value Measures (POVMs), which are a set of operators $\{M_i \}_i$ such that 
\begin{equation}
    M_i \geq 0 ~\forall~i ~ ~ {\rm and} ~ ~  \sum_i M_i = \mathbb{I}.
\end{equation}
Stabilizer circuits are the subset of the stabilizer operations, and consist of Clifford unitaries; preparing and appending stabilizer states; and computational basis measurements. Here, we will consider what POVMs can be implemented via stabilizer circuits.

\subsection{Single Qubit}

For a single qubit, the POVM $\{ \ketbra{0}, \ketbra{1} \}$ captures a measurement in the computational basis. This is a two outcome projective measurement which can alternatively be written as: 
\begin{equation}
    \bigg\{ \Pi_a = \frac{\mathbb{I}+(-1)^a Z}{2} : a \in \{ 0,1 \} ~\bigg\},
\end{equation}
where $Z = \ketbra{0} - \ketbra{1}$ is the Pauli Z operator and $a \in \{0,1\}$ is the single bit output from performing the computational basis measurement --- this corresponds to getting the outcome associated to either $\ketbra{0}$ or $\ketbra{1}$ respectively. 

If a Clifford unitary is applied before a state is measured, it can equivalently be viewed as changing the measurement basis. Applying a unitary  $U \in \mathcal{C}_1$ to a state before measuring in the computational basis is, using the cyclic nature of the trace, equivalent to measuring the POVM
\begin{equation}
    \bigg\{ \Pi^U_a = U^\dagger \bigg(\frac{\mathbb{I}+(-1)^a Z}{2} \bigg)U:~U \in \mathcal{C}_1,~ a \in \{ 0,1 \} ~\bigg\},
\end{equation}
on the initial state. Now, as Clifford unitaries map Pauli operators to Pauli operators, and there exists a Clifford unitary to map each Pauli to every other Pauli, such a POVM can be rewritten as  
\begin{equation}
    \bigg\{ \Pi^P_a = \frac{\mathbb{I}+(-1)^a P}{2} :~P \in \mathcal{P}_1, ~ a \in \{ 0,1 \} ~\bigg\}.
\end{equation}

\subsection{Two to $n$ Qubits}

The above notion can easily be expanded to multiple qubits. We will first do this explicitly for two qubits, before generalising to $n$ qubits. 

Measuring multiple qubits in the computational basis means measuring the single qubit computational basis measurement POVM in each qubit subspace. For two qubits this corresponds to measuring the four outcome POVM
\begin{align} \label{eq: two qubit computational basis measurements}
    &\bigg\{ \Pi_{a,b} = \frac{\mathbb{I}_A+(-1)^a Z_A}{2} \otimes \frac{\mathbb{I}_B+(-1)^b Z_B}{2} : a,b \in \{ 0,1 \} ~\bigg\}, 
\end{align}
where the subscripts $A$ and $B$ label the two qubit spaces with corresponding output bits $a$ and $b$ respectively. Under a pre-processing Clifford unitary, the set of two qubit measurements that can be made under stabilizer circuits is then  
\begin{align}
    \bigg\{ \Pi^{ij}_{a,b} =  \frac{\mathbb{I}_A \otimes \mathbb{I}_B+(-1)^a~P_i + (-1)^b P_j + (-1)^{a \oplus b} P_iP_j}{4} : a,b \in \{ 0,1 \}, ~P_i, P_j \in \mathcal{P}_2, ~[P_i, P_j]=0 ~\bigg\},
\end{align}
where $\oplus$ means addition modulo two. The need for the Pauli's to commute arises due to Clifford unitaries preserving the commuting structure of Pauli strings and $[Z_A \otimes \mathbb{I}_B, \mathbb{I}_A \otimes Z_A]=0$. It can be seen that each $\Pi^{ij}_{a,b}$ is a projector onto a pure stabilizer state --- as it is an equally weighted sum over a mutually commuting set of Pauli-strings. Hence, these POVMs are a projection into a basis of stabilizer states. 

This can easily be extended to an arbitrary number of qubits: on $n$ qubits this corresponds to measuring an observables with an eigenbasis of the form 
\begin{equation*}\label{eq: n-qubit stab measurement}
    \big\{ \Pi^U_a = U^\dagger \ketbra{a} U \big\}_{a}~,
\end{equation*}
where $U \in \mathcal{C}_n$ and $a \in \{0,1\}^n$ is an n-bit-string output from the computational basis measurements. This is again a projection onto a basis of stabilizer states, meaning each POVM element can equivalently be written as an equally weighted sum over mutually commuting Pauli-strings, for example, 
\begin{equation}
    \Pi^U_0 = \frac{1}{2^n} \sum_l P_l : [P_l, P_{l'}] = 0 ~\forall~l, l', P_l \in \mathcal{P}_n. 
\end{equation}
Different basis states, corresponding to different bit-strings $a \in \{0,1\}^n$, then contain the same set of Pauli-string in their decomposition but with different phases of $\{+1,-1\}$ associated to each Pauli string i.e., the same Pauli-strings but different elements of the Pauli group.  

\subsection{$n$ Qubits With Ancilla}

In addition to computational basis measurements and Clifford unitaries, stabilizer circuits also allow for the creation of arbitrary stabilizer states. Here, we consider the additional POVMs made available as a result of this. 

Consider a bipartite space $\mathcal{H}_n^A \otimes \mathcal{H}_k^B$. Via the introduction of a stabilizer state in this $B$ system, the set of available POVMs on the $A$ system can be increased: for a given input state $\rho \in \mathcal{D}(\mathcal{H}_n)$ a $k$ qubit stabilizer state $\omega \in {\rm STAB}(k)$ can first be appended. A joint Clifford unitary can then be applied to the whole space, before a measurement in the computational basis on all $n+k$ qubits is made. Such a measurement is therefore just a projection into a stabilizer basis on $\rho \otimes \omega$, i.e., it is a measurement of the form of Eq.\eqref{eq: n-qubit stab measurement} on all $n+k$ qubits in the $A$ and $B$ spaces. 

Such a measurement will output an $n+k$ bit-strings, with the probability of getting the outcomes $a \in \{0,1\}^{n +k}$ is therefore 
\begin{equation}
    \begin{split}
        p_{a} &= \textrm{tr}\big[ U^\dagger \ketbra{a} U (\rho \otimes \omega) \big] \\
        &= \textrm{tr}\big[ M^\omega_{a} \rho \big],
    \end{split}
\end{equation}
where $M^\omega_{a} = \textrm{tr}_B \big[  U^\dagger \ketbra{a} U (\mathbb{I} \otimes \omega) \big]$ is an element of a POVM acting on just the $A$ space, although, it depends on the system in the $B$ space. For a given $\omega$, one is therefore able to measure a POVM ${M^\omega_{a}}$ on the $A$ system. 
  
\subsection{Adaptive Measurements} \label{appendix: stabilzer measurements, adaptive measurements}

The above sets of measurements achievable with stabilizer circuits are fixed measurements i.e., they are determined before the circuit begins. We also consider the possibility of the player performing {\em adaptive stabilizer circuits}. Under adaptive circuits, measurement outcomes on part of the system determine which subsequent operations (measurements) are applied. 
Specifically, we consider a bipartite space $\mathcal{H}^A_n \otimes \mathcal{H}_k^B$, where a fixed measurement on the $B$ space is used to determine which of a number of fixed measurement is then performed on the $A$ space. With $\omega_B \in {\rm STAB}(k)$ in the ancilla, one can model adaptive measurements on a state $\rho \in \mathcal{D}(\mathcal{H}_n)$ as a quantum instrument of the form 
\begin{equation}
    \mathcal{I}(\rho) = \sum_{a \in \{0,1\}^n} \sum_{b \in \{0,1\}^k} (\ketbra{a} \otimes \mathbb{I}_A) (V^b_A \otimes \mathbb{I}_B)(\mathbb{I} \otimes \ketbra{b})U(\rho \otimes \omega_B)U^\dagger(\mathbb{I} \otimes \ketbra{b})(V^{b,\dagger}_A \otimes \mathbb{I}_B)(\ketbra{a} \otimes \mathbb{I}_A) \otimes \ketbra{ab},
\end{equation}
such that depending on the measurement output from the $B$ system i.e., the bit-string $b \in \{0,1\}^k$, there is a Clifford unitary $V^b \in \mathcal{C}_n$ applied to the $A$ system. A computational basis measurement on the $A$ system is then applied, outputting a bit-string $a \in \{0,1\}^n$. Both bit-strings are stored in a classical register.  

The probability of getting the output bit-strings $a \in \{0,1\}^n$ and $b \in \{0,1\}^k$ is therefore given by 
\begin{equation}
    \begin{split}
        p_{(a,b)} &= \textrm{tr}\big[ U^\dagger \big( V^{b,\dagger}_A \ketbra{a} V^b_A \otimes \ketbra{b} \big) U (\rho \otimes \omega_B) \big] \\
        &= \textrm{tr}\big[ T^{U,\omega, \{V\}}_{a,b} \rho \big]
    \end{split}
\end{equation}
where 
\begin{equation}
    T^{U, \omega,  \{V\}}_{a,b} = \textrm{tr}_B\big[ U^\dagger \big( V_A^{b,\dagger} \ketbra{a} V_A^b \otimes \ketbra{b} \big) U (\mathbb{I}_A \otimes \omega_B) \big]
\end{equation}
are elements of a POVM on the $A$ space, with $U \in \mathcal{C}_{n+k}, V^b \in \mathcal{C}_n~\forall~b$. This POVM depends on the state in the $B$ system, the global Clifford unitary, and the choice of Clifford unitaries $\{ V^b \}_b$ applied to the $A$ system.  
\section{Non-stabilizer measurements}
\subsection{ n-qudits minimum error success probability via an SDP}\label{sm: SDP for non-stabilizer measurements}
Here we consider the problem of  determining the minimum error success probability of discriminating correctly between $2$ states of $n$ qudits, $\rho_0$ and $\rho_1$, given that we allow the measurements to have a maximum value of non-stabilizerness $\mu\in [0,\mu_{max}]$. This generalizes the stabilizer measurements, corresponding to $\mu=0$.  We show that it can be written as a semidefinite program, allowing in principle to consider an arbitrary finite number of states and arbitrary prime dimension $d$. As we will see, our approach is subjected to the exponential growth of the stabilizer polytope, so even if stated as an SDP, we do not claim an efficient calculation for a large number of systems or high dimensions. Setting $\Delta =\rho_0-\rho_1$, the QSD problem can be framed as the following optimization
\begin{align}
    &\max_{M^\mu_0, M^\mu_1} \frac{1}{2} + \frac{1}{2} Tr(\Delta M^\mu_0),\label{eq:opt_problem_M_mu}\\
    &s.t.\nonumber\\
    &\quad\mathcal{M}(M^\mu_i)\leq \mu \quad for \quad i= 0,1, \nonumber\\
    &\quad M^\mu_i\succeq 0,\nonumber\\
    &\quad M^\mu_0+M^\mu_1 =\mathbb{I}\nonumber,\\
\end{align}
specifying that both elements ${M^\mu_0, M^\mu_1}$ form a POVM and that each of them is restricted to have a maximum value $\mu$ of non-stabilizerness respect to some monotone $\mathcal{M}$.  In particular, we choose the measure defined in~\cite{Cao_2025, palhares_2026} as:
\begin{equation}
    \mathcal{M}(\rho) = \min_{\sigma \in \mathrm{STAB}} \frac{1}{2} \| \rho - \sigma \|_1,
\end{equation}
which is the trace distance between $\rho$ and the stabilizer polytope, vanishing if and only if $\rho$ is a convex combination of stabilizer pure states.
Note that the definition is valid only for states $\rho$, which are normalized. Therefore, to extend it to our problem, for $\text{Tr}(M_i^\mu)\neq 0$ we consider:
\begin{equation}
    \mathcal{M}\left(\frac{M^\mu_i}{\text{Tr}(M^\mu_i)}\right) = \min_{\sigma_i \in \mathrm{STAB}} \frac{1}{2} \| \frac{M^\mu_i}{\text{Tr}(M^\mu_i)} - \sigma_i \|_1\leq\mu, \quad\quad for\quad i=0,1.
\end{equation}
Note that we have defined different optimization variables $\sigma_i$, one for each POVM element $i$. Since $M^\mu_i\succeq0$ for $i=0,1$, we can rewrite this condition as:  
\begin{equation}
    \min_{\sigma \in \mathrm{STAB}}  \| M^\mu_i -  \text{Tr}(M^\mu_i)\sigma_i\|_1\leq2\mu\text{Tr}(M^\mu_i). \label{eq:Magic_constraint_M}
\end{equation}

In what follows we show that this restriction can be linearized and the optimization~\eqref{eq:opt_problem_M_mu} can be stated as a semidefinite program. As it is well known~\cite{skrzypczyk2023, palhares_2026}, the 1-norm of an Hermitian operator $A$ admits the following SDP representation.
\begin{align}
    & \|A\|_1 = \min_{X} \text{Tr}(X),\nonumber\\
    & s.t.\nonumber\\
    &\quad\quad X\succeq0 \label{eq:1norm_C1}\\
    &\quad \quad X\succeq A\succeq-X \label{eq:1norm_C2},
\end{align}
where $X$ is an auxiliary positive semidefinite matrix.
Therefore, we have the following chain of inequalities:
\begin{equation}
    \min_{\sigma_i \in \mathrm{STAB}}  \| M^\mu_i -  \text{Tr}(M^\mu_i)\sigma_i\|_1 \leq \| M^\mu_i -  \text{Tr}(M^\mu_i)\sigma_i^{guess}\|_1  = \min_{X_i \textit{ s.t. (\ref{eq:1norm_C1}),(\ref{eq:1norm_C2})}} \text{Tr}(X_i) \leq \text{Tr}(X_i^{guess}),\label{eq: Chain inequalities_SDP}
\end{equation}
where the first inequality is true because any random guess $\sigma_i^{guess}$ would be further away from $M^\mu_i/\text{Tr}(M^\mu_i)$ than the closest stabilizer state $\sigma_i^*=\arg\min_{\sigma_i \in \mathrm{STAB}}  \| M^\mu_i -  \text{Tr}(M^\mu_i)\sigma_i\|_1 $. The second equality is the definition of the 1-norm as an SDP, so the minimization over $X_i$ is restricted to  $X_i\succeq M^\mu_i -  \text{Tr}(M^\mu_i)\sigma_i^{guess}\succeq-X_i$ and $X_i\succeq0$. Finally, the last inequality shows that any random $X_i^{guess}$ would have a higher trace than $X_i^*=\arg\min \text{Tr}(X_i)$. It is assumed that both $X_i^{guess}$ and $X_i^*$ are subjected to the same constraints.
Eq.~(\ref{eq: Chain inequalities_SDP}) implies Eq.~(\ref{eq:Magic_constraint_M}) if we set
\begin{align}
    &  \text{Tr}(X^{guess})\leq 2\mu\text{Tr}(M^\mu_i),\nonumber\\
    & s.t.\nonumber\\
    &\quad \quad X^{guess}\succeq M^\mu_i - \sigma_i'\succeq-X^{guess},
\end{align}
which is a linear constriant. Here we have defined the positive semidefinite but unnormalized state $\sigma_i' =\sigma_i^{guess} \text{Tr}(M_i)$. 

Now note that we can constrain $\sigma_i^{guess}$ to be a stabilizer state by a set of linear restrictions given by the operators $\{A_q\}_q$ defining the facets of the stabilizer polytope~\cite{Howard2014}:
\begin{equation}
    \sigma_i^{guess}\in \mathrm{STAB}_{d,n} \iff \text{Tr}(\sigma_i^{guess} A_q)\geq 0 \quad \forall q.
\end{equation}

Note that this condition implies that $\text{Tr}(\sigma_i^{guess}) =1$, however given the unnormalized state $\sigma_i'$ we must include the condition $\text{Tr}(\sigma_i') = \text{Tr}(M_i)$.

Putting all together, and relabeling $X_i^{guess}=X_i$, the optimization problem~(\ref{eq:opt_problem_M_mu}) can be written as
\begin{align}
    &\max_{\{M^\mu_i, X_i, \sigma'_i\}_i} \frac{1}{2} + \frac{1}{2} Tr(\Delta M^\mu_0),\label{eq:SDP_M_mu}\\
    &s.t.\nonumber\\
    &\quad\text{Tr}(X_i)\leq 2\mu\text{Tr}(M^\mu_i),\nonumber\\&\quad X_i\succeq M^\mu_i - \sigma_i'\succeq-X_i, \nonumber\\
    &\quad M^\mu_i\succeq 0,\nonumber\\
    &\quad M^\mu_0+M^\mu_1 =\mathbb{I},\nonumber\\
    &\quad\text{Tr}(\sigma_i' A_q)\geq 0 \quad \forall q,\nonumber\\
    &\quad \text{Tr}(\sigma_i') = \text{Tr}(M^\mu_i)\nonumber\\
    &\quad\sigma_i'\succeq 0,
\end{align}
for $\quad i\in\{0,1\}$. Note that the facet inequalities naturally enforce the positive semidefiniteness of $\sigma_i'$ so the last constraint can be removed, however we include it to show explicitly that the resulting optimizaion is an SDP. We emphasize that although semidefinite programs can be computed in polynomial time, the constraints include the restriction of $\sigma_i'$ to satisfy all the inequalities defining the facets of the polytope. Since these grow super exponentially, the problem becomes quickly untractable for more than a few qubits or dimensions. However, the solution is general, and in particular it allows to obtain the curves shown in Fig.~\ref{figure:Psucc_mu} of the main text. 

\section{When does Non-stabilizerness Provide an Advantage Under Adaptive Stabilizer Circuits}\label{appendix: When does Non-stabilizerness provide an advantage}

From Result~\ref{result: appending any qubit state with feedfoward} it can be seen that having an ancillary state proves advantageous for performing QSD under adaptive stabilizer circuits, over fixed stabilizer circuits, if
\begin{equation}\label{eq: condition for advanatge in appendix}
    2 \vert \vert \Delta \vert \vert_\infty < \vert \vert \Delta \vert \vert_{2-{\rm ky}} \vert \vert \vec n \vert \vert_{2-{\rm ky}}.
\end{equation}
Now, let 
\begin{equation}
    \vert \vec \Delta \vert^\downarrow = \begin{pmatrix}
        \vert \Delta_1 \vert \\
         \vert \Delta_2 \vert \\
          \vert \Delta_3 \vert 
    \end{pmatrix}, 
\end{equation}
where $ \vert \vec v \vert^\downarrow$ is the absolute value of the vector $v$ ordered in non-increasing order, such that $\vert \Delta_1 \vert \geq \vert \Delta_2 \vert \geq \vert \Delta_3 \vert$. Using this, Eq.~\eqref{eq: condition for advanatge in appendix} becomes 
\begin{equation}
    \frac{2-\vert \vert \vec n \vert \vert_{2-{\rm ky}}}{\vert \vert \vec n \vert \vert_{2-{\rm ky}}} < \frac{\vert \Delta_2 \vert}{\vert \Delta_1 \vert}.
\end{equation}
As the success probability for QSD under adaptive stabilizer circuits is maximised if $\vert \vert \vec n \vert \vert_{2-{\rm ky}}=\sqrt{2}$, meaning that $\vec n$ is the Bloch vector of an $H$-type magic state, then a sufficient condition for non-stabilizerness to {\em not} provide an advantage is
\begin{equation}
        \frac{2-\sqrt{2}}{\sqrt{2}} \geq \frac{\vert \Delta_2 \vert}{\vert \Delta_1 \vert}.
\end{equation}
See Fig.~\ref{appendix figure of region where nonstabs helps} for a plot of the regions of the $xz$-plane of $\vec \Delta$ for which an $H$-type magic state does and does not facilitate an advantage in QSD under adaptive stabilizer circuits. The correlation between the segments in Fig.~\ref{appendix figure of region where nonstabs helps} and Fig.~\ref{figure: semi-circle figure} is clear to see. We note, if one uses a general state, rather than an $H$-type magic state, in the ancillary system then the angle of the region in which an advantage is gained (green region in Fig.~\ref{appendix figure of region where nonstabs helps}) will get smaller. See Fig.~\ref{not H} and Fig.~\ref{appendix figure of region where nonstabs close to stab} for examples. 
\begin{figure}
    \centering
    \includegraphics[scale=0.5]{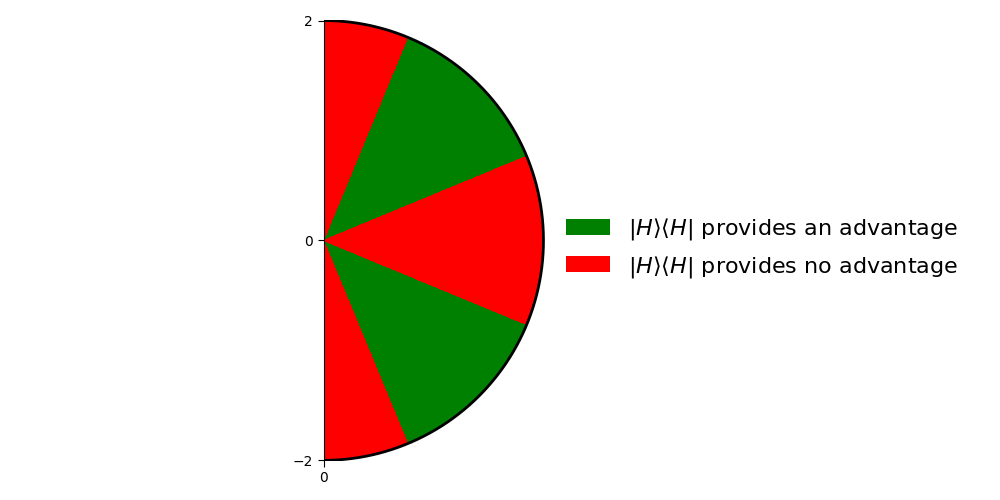}
    \caption{The regions for which access to an $H$-type magic state does and does not allow an advantage to be gained in QSD under adaptive stabilizer circuits for all $\Vec{\Delta}$ in the $xz$-plane.}\label{appendix figure of region where nonstabs helps}
\end{figure}

\begin{figure}
    \centering
    \includegraphics[scale=0.5]{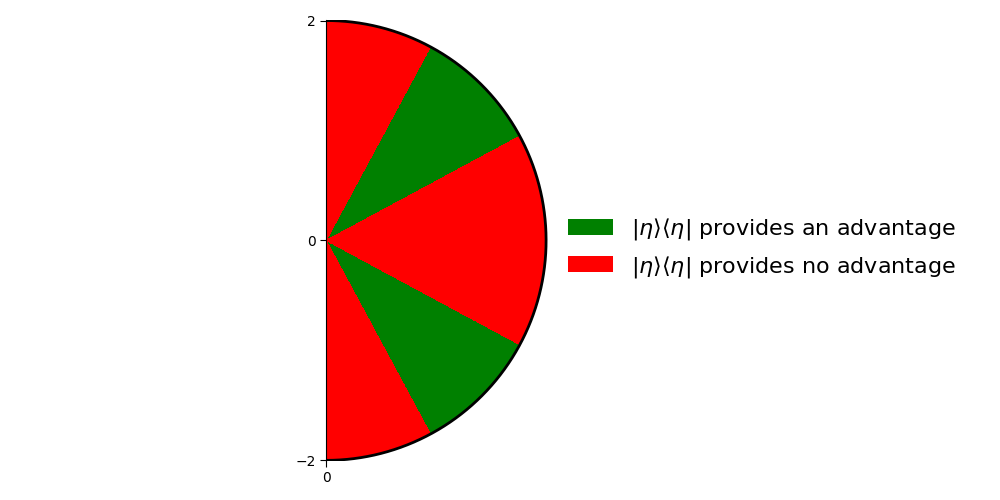}
    \caption{The regions for which access to a state $\ket{\eta}$ which has $\vert \vec n \vert^\downarrow = (0.9, 0.4, 0.173)$ does and does not allow an advantage to be gained in QSD under adaptive stabilizer circuits for all $\Vec{\Delta}$ in the $xz$-plane.}\label{not H}
\end{figure}

\begin{figure}
    \centering
    \includegraphics[scale=0.5]{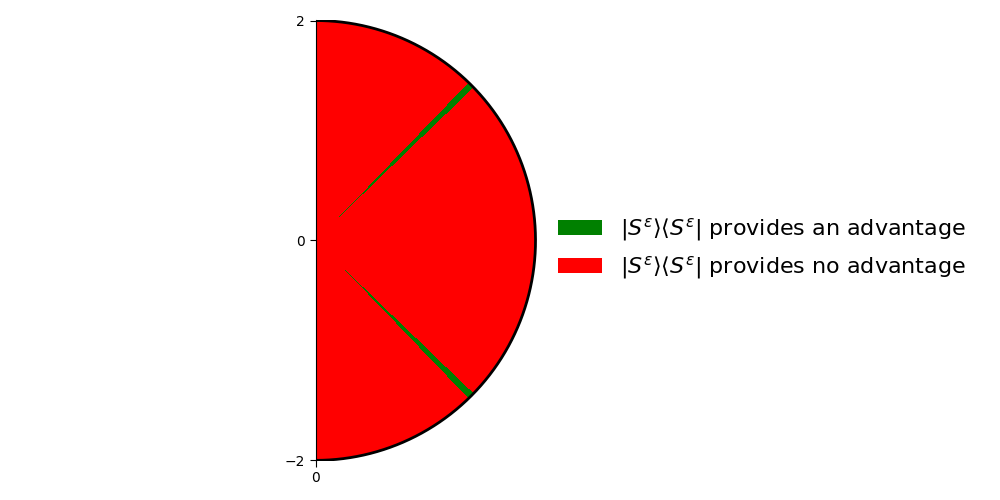}
    \caption{The regions for which access to a state $\ket{S^\epsilon}$ which has Bloch vector $\vert \vec n\vert^\downarrow  = (1-\epsilon, \sqrt{\epsilon(2-\epsilon)}, 0)$, which is close to a pure stabilizer state, does and does not allow an advantage to be gained in QSD under adaptive stabilizer circuits for all $\Vec{\Delta}$ in the $xz$-plane.}\label{appendix figure of region where nonstabs close to stab}
\end{figure}

\section{Proofs}

\subsection{Proof of Result~\ref{theorem: STABs don't increase sucess prob}} \label{appendix: lemma: STABs don't increase sucess prob}

Before beginning the proof, we re-derive some known proofs from the literature used in the proof of Result~\ref{theorem: STABs don't increase sucess prob} for completeness. 
\\
\\
We now confirm a known result that projecting into a subspace of a pure state with a pure state leaves a pure state up to a constant $\alpha \in [0,1]$ in the remaining system. 
\begin{lemma}\label{appendix lemma: projecting into pure states}
    Consider a bipartite space $\mathcal{H}^A \otimes \mathcal{H}^B$, and let $\ketbra{\Phi} \in \mathcal{D}(\mathcal{H}^A \otimes \mathcal{H}^B)$ whilst $\ketbra{\Psi} \in \mathcal{D}(\mathcal{H}^B)$, such that each state is a pure state in their respective spaces. Then,
    \begin{equation}
         \textnormal{tr}_B \big[ \ketbra{\Phi}_{AB} (\mathbb{I}_A \otimes \ketbra{\Psi}_B ) \big] = \alpha \ketbra{\Tilde{\chi}},
    \end{equation}
    where $\ketbra{\Tilde{\chi}} \in \mathcal{D}(\mathcal{H}^A)$ and $\alpha \in [0,1]$. 
\end{lemma}
\begin{proof}
    Using the Schmidt decomposition of $\ketbra{\Phi}_{AB}$, it follows that 
    \begin{equation}
        \begin{split}
            \textrm{tr}_B \big[ \ketbra{\Phi}_{AB} (\mathbb{I}_A \otimes \ketbra{\Psi}_B ) \big] &= \sum_{i,j} \sqrt{\lambda_i \lambda_j} ~ \textrm{tr}_B \big[ \ketbratwo{ii}{jj}(\mathbb{I}_A \otimes \ketbra{\Psi}) \big] \\
            &= \sum_{i,j} \sqrt{\lambda_i \lambda_j}~ \ketbratwo{i}{j} ~\bigg( \braket{j \vert \Psi}\braket{\Psi \vert i} \bigg). \\
        \end{split}
    \end{equation}
    One can then define 
    \begin{equation}
        \ket{\chi} = \sum_i \sqrt{\lambda_i} \braket{\Psi \vert i} \ket{i},
    \end{equation}
    such that 
    \begin{equation}
         \textrm{tr}_B \big[ \ketbra{\Phi}_{AB} (\mathbb{I}_A \otimes \ketbra{\Psi}_B ) \big] = \ketbra{\chi}.
    \end{equation}
   It can be seen that 
    \begin{equation}
        \alpha = \textrm{tr}\big[ \ketbra{\chi} \big] = \sum_{i} \vert \braket{\Psi | i} \vert^2 \lambda_i \leq 1,
    \end{equation}
    where the inequality comes from fact that $0 \leq \vert  \braket{\Psi | i} \vert^2 \leq 1 ~\forall~i$, and $0 \leq \lambda_i \leq 1~\forall~i$ as each $\lambda_i$ are eigenvalues of a density operator ($\ketbra{\Phi}_{AB}$). Hence,
    \begin{equation}
        \ketbra{\Tilde{\chi}} = \frac{1}{\alpha} \ketbra{\chi},
    \end{equation}
    is a valid pure quantum state, such that 
    \begin{equation}
        \alpha \ketbra{\Tilde{\chi}} = \textrm{tr}_B \big[ \ketbra{\Phi}_{AB} (\mathbb{I}_A \otimes \ketbra{\Psi}_B ) \big],
    \end{equation}
    completing the proof. 
\end{proof}

We now use Lemma~\ref{appendix lemma: projecting into pure states} to re-derive the fact that projecting a pure stabilizer state into the subspace of another pure stabilizer will leave a pure stabilizer state in the remaining subspace. 
\begin{lemma}\label{Lemma: projecting into stab states gives stab states}
    Let $\ketbra{\Phi} \in \mathcal{D}(\mathcal{H}_n^A \otimes \mathcal{H}_k^B)$ and $\ketbra{\Psi} \in \mathcal{D}(\mathcal{H}_k^B)$ be pure stabilizer states. Then 
    \begin{equation}
         \textnormal{tr}_B \big[ \ketbra{\Phi}_{AB} ( \mathbb{I}_A \otimes \ketbra{\Psi}_B) \big] = \alpha \ketbra{\Tilde{\chi}}
    \end{equation}
    where $\ketbra{\Tilde{\chi}} \in {\rm STAB}(n)$ and $\alpha \in [0,1]$, i.e., the output is proportional to a pure stabilizer state. 
\end{lemma}
\begin{proof}
As $\ketbra{\Phi}_{AB}$ is a stabilizer state, its Pauli decompositions is 
\begin{equation}
    \ketbra{\Phi}_{AB} = \frac{1}{2^{n+k}} \sum_l c_l P_l : [P_l, P_{l'}] = 0~\forall~l,l',~ P_l \in \mathcal{P}_{n+k},
\end{equation}
such that each $P_l$ is a $n+k$ qubit Pauli-string and $c_l \in \{-1,+1\}$. The same is therefore true of $\ketbra{\Psi}_B$, i.e.,  
\begin{equation}
    \ketbra{\Psi}_B = \frac{1}{2^k} \sum_j d_j P_j : [P_j, P_{j'}] = 0~\forall~j,j', ~P_j \in \mathcal{P}_k, ~d_j \in \{+1, -1\}.
\end{equation}
We then note that each $n+k$ qubit Pauli-string can be written as the composition of an $n$ and $k$ qubit Pauli string:
\begin{align}
    P_l = H_l^A \otimes H^B_l: ~ ~ H^A_l \in \mathcal{P}_n, ~ ~ H^B_l \in \mathcal{P}_k~\forall~l.
\end{align}
As $[P_l, P_{l'}]=0~\forall~l,l'$ it must be the case that 
\begin{align}\label{eq: Pauli-conditons for local commuting}
    [H^A_l, H^A_{l'}] = 0 ~ ~ &{\rm and} ~ ~ [H^B_l, H^B_{l'}] = 0, ~ ~ {\rm or}\\
    \{H^A_l, H^A_{l'}\} = 0 ~ ~ &{\rm and} ~ ~ \{H^B_l, H^B_{l'}\} = 0.
\end{align}
Now, it holds that 
\begin{equation}
    \begin{split}
        \textrm{tr}_B \big[ \ketbra{\Phi}_{AB} ( \mathbb{I}_A \otimes \ketbra{\Psi}_B) \big] = \frac{1}{2^{n+k}} \frac{1}{2^k} \sum_{jl} (c_l d_j) H^A_l \textrm{tr}\big[ H^b_l P_j \big], 
    \end{split}
\end{equation}
where 
\begin{equation}
    \textrm{tr}\big[ H^b_l P_j \big] = \begin{cases}
        2^k & {\rm if } ~ H^B_l = P_j \\
        0 & {\rm otherwise}
    \end{cases}
\end{equation}
and $(c_ld_j) \in \{-1, +1\}~\forall~l,j$. Hence, only the set of Pauli strings $\{ H^A_l \otimes P_j \}_{l,j}$ are non-zero after the Partial trace. Therefore, due to the fact that all $P_j$'s are mutually commuting and Eq.\eqref{eq: Pauli-conditons for local commuting}, all remaining Pauli-strings in the $A$ space must mutually commute. From here, it can be seen that 
\begin{equation}
    \begin{split}
        \textrm{tr}_B \big[ \ketbra{\Phi}_{AB} ( \mathbb{I}_A \otimes \ketbra{\Psi}_B) \big] &= \frac{1}{2^{n+k}} \sum_{r} \kappa_r C_r H^A_r : [H^A_r, H^A_{r'}]=0~\forall~r,r', \\
    \end{split}
\end{equation}
where $\kappa_r \geq 0$ is a constant that counts the number of times each string $H_r^A$ appears and $C_r \in \{-1, +1\}$. To find $\kappa_r$, we first define the sets $\mathfrak{S}_{AB}$ and $\mathfrak{S}_B$, which are the Pauli-strings in the decomposition of $\ketbra{\Phi_{AB}}$ and $\ketbra{\Psi}_B$ respectively i.e., the stabilizers of the state up to some phase on each Pauli. We then define 
\begin{equation}
    \mathfrak{S}_{IB} = \big\{ \mathbb{I} \otimes P_j : \mathbb{I} \otimes P_j \in \mathfrak{S}_{AB}, ~P_j \in \mathfrak{S}_B \big\},
\end{equation}
which are the stabilizers of $\ketbra{\Phi}_{AB}$ that act trivially on the $A$ system. As the stabilzers of a state form a closed group, if $H^A \otimes P_{j'} \in \mathfrak{S}_{AB}$ and $\mathbb{I} \otimes P_j \in \mathfrak{S}_{IB}$ then $H_A \otimes P_jP_{j'} \in \mathfrak{S}_{AB}$. Hence, $\kappa_r \leq \vert \mathfrak{S}_{\mathbb{I}B} \vert \leq 2^k$, as there can be a maximum of $2^k$ commuting Pauli-strings in a $k$ qubit space. More specifically, if ${c}_j$ are the phases of the Pauli-strings in $\mathfrak{S}_{IB}$, then 
\begin{equation}
    \kappa_r = \big\vert \sum_j c_j \big\vert,
\end{equation}
meaning $\kappa=\kappa_r~\forall~r$ i.e, it is independent of the specific Pauli-string in the decomposition. Hence, 
\begin{equation}
    \begin{split}
        \textrm{tr}_B \big[ \ketbra{\Phi}_{AB} ( \mathbb{I}_A \otimes \ketbra{\Psi}_B) \big] &= \frac{\kappa}{2^{n+k}} \sum_{r} C_r H^A_r : [H^A_r, H^A_{r'}]=0~\forall~r,r', \\
        &= \frac{\kappa}{2^k} \bigg( \frac{1}{2^n} \sum_r C_r H_r^A \bigg) \\
        &= \frac{\kappa}{2^k} ~ 2^{m-n} \Pi_G,
    \end{split}
\end{equation}
where  $G$ is an abelian subgroup of the $n$-qubit Pauli strings and $\Pi_G = \frac{1}{2^m} \sum_r C_r H_r^A$ is the rank $2^{n-m}$ projector onto the joint $+1$ eigenspace of the group.

Now, from Lemma~\ref{appendix lemma: projecting into pure states} it is known that this is a pure state up to a constant that is less that or equal to $1$, meaning it is at most rank 1. Therefore, $\Pi_G$ must be rank $1$, which occurs if $m=n$, meaning that $G$, the abelian subgroup of the $n$-qubit Pauli strings, is maximal. Hence, this is a pure stabilizer state up to a constant $\kappa/2^k$ where $0 \leq \kappa \leq 2^k$. 
\end{proof}
We now restate Result~\ref{theorem: STABs don't increase sucess prob} before proving it. 
\\
\\
{\em {\bf Result}~\ref{theorem: STABs don't increase sucess prob}.~
    $ P_{\rm suc}^{\rm stab}(\Delta) = P_{\rm suc}^{\rm stab, f}(\Delta, \omega) ~\forall~\omega \in {\rm STAB}$. }
\begin{proof}
    Here, we consider a space $\mathcal{H}_n \otimes \mathcal{H}_k$, such that $\Delta$ is the difference of two $n$-qubit density operators, and $\omega$ is a stabilizer state on $k$-qubits. As mentioned in the main text, without loss of generality we can set $\omega=\ketbra{\omega} \in {\rm STAB}$. 

    Under fixed stabilizer circuits POVM's of the form 
    \begin{equation*}
    \big\{ M^{U, \omega}_a = \textrm{tr}_B \big[ U^\dagger \ketbra{a} U (\mathbb{I}_A \otimes \ketbra{\omega}_B) \big] \big\}_{a}~,
    \end{equation*}
    are considered. Given $U^\dagger \ketbra{a} U$ is a pure stabilizer state on $n+k$ qubits, as $\ketbra{a} \in STAB(n)$ and $U \in \mathcal{C}_{n+k}$, it can be seen from Lemma~\ref{Lemma: projecting into stab states gives stab states} that 
    \begin{equation}
        M^{U, \omega}_a = \textrm{tr}_B \big[ U^\dagger \ketbra{a} U (\mathbb{I}_A \otimes \ketbra{\omega}_B) = \frac{\kappa_a}{2^k} \ketbra{\eta_a}_A,
    \end{equation}
    where $\ketbra{\eta_a}$ is a pure stabilizer state where the Pauli strings in its Pauli decomposition are independent of $a$, but the sign of those strings are dependent on $a$. This is due to all $\ketbra{a}$ having the same Pauli-strings in its Pauli-decomposition for all $a$, but with different signs. Moreover, $0 \leq \kappa_a \leq 2^k$ is a constant that also depends on $a$. Hence, all $\ketbra{\eta_a}$ are elements of the same stabilizer basis. Now, as $\sum_a M^{U, \omega}_a = \mathbb{I}_A$, and letting this stabilizer basis be $\{ \ketbra{\Tilde{\eta}_s}\}_s$, such that $s \in \{0,1\}^n$ labels each of the $2^n$ different stabilizer basis states, there must exists subsets $A \subset \{0,1\}^{n+k}$ such that 
    \begin{equation}
        \ketbra{\Tilde{\eta}_s} = \sum_{a \in A} \frac{\kappa_a}{2^k} \ketbra{\eta_a},
    \end{equation}
    where $\ketbra{\eta_a} = \ketbra{\eta_s}~\forall~a \in A$.
    Letting $\mathfrak{A}$ be the set of all such subsets, and assuming the optimal unitary has been used, the minimum error success probability when using a POVM of this form becomes 
    \begin{equation}
        \begin{split}
            P^{\rm stab, f}_{\rm suc}(\Delta, \omega) &= \frac{1}{2} + \frac{1}{4} \sum_{a \in \{0,1\}^{n+k}} \big\vert \textrm{tr} \big[ M^{U,\omega}_a (\Delta) \big] \big\vert, \\
            &= \frac{1}{2} + \frac{1}{4} \sum_{A \in \mathfrak{A}} \sum_{a \in A } \bigg \vert \frac{\kappa_a}{2^k} \textrm{tr} \big[ \ketbra{\eta_a} (\Delta) \big] \bigg \vert, \\
            &= \frac{1}{2} + \frac{1}{4}\sum_{s \in \{0,1\}^n } \bigg \vert \textrm{tr} \big[ \ketbra{\Tilde{\eta}_s} (\Delta) \big] \bigg \vert, \\
            & \leq P_{\rm suc}^{\rm stab}(\Delta).
        \end{split}
    \end{equation}
    To complete the proof, we note that equality can always be achieved by ignoring the ancillary system and performing a unitary of the form $U_{AB} = U_A \otimes \mathbb{I}_B \in \mathcal{C}_{n+k}$, where $U_A \in \mathcal{C}_n$.
\end{proof}

\subsection{Proof of Lemma~\ref{lemma: qubit fixed stab measurement no ancilla}} \label{appendix: proof of qubit fixed stab measurement no ancilla}
Let $\rho, \sigma \in \mathcal{D}(\mathcal{H}_1)$ and $\Vec{\Delta} \coloneq n_\rho - n_\sigma$ where $n_\rho$ and $n_\sigma$ are the Bloch vectors of $\rho$ and $\sigma$ respectively.
\\
\\
{\em {\bf Lemma~\ref{lemma: qubit fixed stab measurement no ancilla}}
        $P_{\rm suc}^{\rm stab}(\vec \Delta) = \frac{1}{2} + {\frac{1}{4}} \vert \vert \Vec{\Delta} \vert \vert_\infty,$
    where $\vert \vert \Vec{\Delta} \vert \vert_\infty$ is the vector infinity norm i.e., the largest component of the absolute value of $\Vec{\Delta}$.
}

\begin{proof}
    When restricted to qubit stabilizer circuits, only a measurement in the eigenbasis of $Z$ following a pre-processing Clifford unitary $U \in \mathcal{C}_1$ can be applied. This is equivalent to measuring a POVM of the form \{$\ketbra{\psi}, \mathbb{I} - \ketbra{\psi} \}$ where $\ket{\psi} \in {\rm STAB(1)}$. With this in mind, the minimal error success probability becomes
    \begin{align}
        P^{\rm stab}_{\rm suc} (\Delta) &= \frac{1}{2} + \max_{U \in \mathcal{C}_1} \frac{1}{4} \sum_{a \in \{0,1\}} \big\vert \textrm{tr} \big[ \Pi_a^U (\Delta) \big] \big\vert \\
        &= \frac{1}{2} + \frac{1}{2}~\max_{\ket{\psi} \in STAB} \big\vert \textrm{tr}\big[ \ketbra{\psi} \Delta \big] \big\vert \\
        &= \frac{1}{2} + \frac{1}{4}~\max_{\ket{\psi} \in STAB} \bigg\vert\textrm{tr}\big[\Delta X \big] \textrm{tr}\big[ \ketbra{\psi} X \big] + \textrm{tr}\big[\Delta Y \big] \textrm{tr}\big[ \ketbra{\psi} Y \big] +  \textrm{tr}\big[\Delta Z\big] \textrm{tr}\big[ \ketbra{\psi} Z \big] \bigg\vert   \label{eq: final optimisation}
    \end{align}
    where in the final line we have decomposed $\Delta$ into its Pauli basis such that $\Delta = \Delta_x X + \Delta_y Y + \Delta_z Z$, where $\Delta_x = \textrm{tr}\big[\Delta X \big]$, for example. The pure qubit stabilizer states are then 
    \begin{equation}
        {\rm STAB}(1) = \big\{ \ket{0}, \ket{1}, \ket{+}, \ket{-}, \ket{y^+}, \ket{y^-} \big\}, \label{eq: STAB states}
    \end{equation}
    which are the eigenstates of the qubit Pauli operators. It can then be seen that if $\ketbra{\psi} =\frac{1}{2}(\mathbb{I} \pm \kappa_j) \in {\rm STAB}(1)$, where $\kappa_1=X, \kappa_2=Y, \kappa_3=Z$, then 
    \begin{align}
        \textrm{tr}\big[ \ketbra{\psi} \kappa_i \big] &= \frac{1}{2} \textrm{tr}\big[ (\mathbb{I} \pm \kappa_j) \kappa_i \big] \\
        &= \pm \delta_{ij},
    \end{align}
    as $\textrm{tr}\big[\kappa_i] = 0 ~ \forall ~ i$ and $\textrm{tr}\big[\kappa_i \kappa_j \big]=2\delta_{ij}$. From here, it can be seen that Eq.~\eqref{eq: final optimisation} is maximised by choosing a qubit stabilizer state associated to the Pauli-coefficient of $\Delta$ with the largest absolute value. Hence, 
       \begin{align}
        P_{\rm suc}^{\rm stab}(\Delta) &= \frac{1}{2} + {\frac{1}{4}} \max_{X,Y,Z} \big\{ |\textrm{tr}\big[\Delta X \big]|, |\textrm{tr}\big[\Delta Y \big]| , |\textrm{tr}\big[\Delta Z \big]|  \big\} \\
        &= \frac{1}{2} + {\frac{1}{4}} \vert \vert \Vec{\Delta} \vert \vert_\infty,
    \end{align}
    where in the final line we have introduced $\vec \Delta \coloneq \vec n_\rho - \vec n_\sigma$, where $\vec n_\rho$ and $\vec n_\sigma$ are the Bloch vector of $\rho$ and $\sigma$ respectively. Hence, $\textrm{tr}\big[\Delta \kappa_i\big] = (\vec \Delta)_i$, where $(\vec \Delta)_i$ is the component of $\vec \Delta$ associated to the Pauli operator $\kappa_i$. This completes the proof. 
\end{proof}

\subsection{Proof of Result~\ref{result: for qubits, not about of resource is useful}} \label{appendix: result: for qubits, not about of resource is useful}

Firstly, we prove a Lemma for the minimum error success probability under fixed measurements when \hbox{$\rho, \sigma \in \mathcal{D}(\mathcal{H}_2)$}, such that 
\begin{equation}
    \Delta \coloneq \rho - \sigma = \sum_{k} \Delta_{k} P_k, ~ ~ ~ P_k \in \mathcal{P}_2~\forall~i,j. 
\end{equation}
In this case, the figure of merit we are considering is 
\begin{equation}
    \begin{split}
        P_{\rm suc}^{\rm stab, f}(\Delta) &= \frac{1}{2} + \frac{1}{4} \max_{U \in \mathcal{C}_2}  \sum_{a \in \{0,1\}^1} \sum_{b \in \{0,1\}^1} \big\vert \textrm{tr} \big[ (U^\dagger \ketbra{ab} U) \Delta \big] \big\vert, \\
    \end{split}
\end{equation}
We note that for two qubits, a stabilizer basis is characterised by two Pauli-strings, $P_i$ and $P_j$, such that the Pauli-decomposition of each of the basis states contains only $\mathbb{I}, P_i, P_j$ and $P_iP_j$ with different phases. Once this expression is known, we then apply it to an input state of the form $(\cdot) \otimes \omega$. 
\\
\\
{\bf Two Qubits Fixed Measurements}

\begin{lemma}\label{lemma: appendix : two qubit min error sucess prob}
    The minimum error success probability with fixed measurement on two qubits is  
    \begin{equation} \label{result: stab min error equality}
        P_{\rm suc}^{\rm stab, f}(\Delta) = \frac{1}{2} + \frac{1}{2} \max_{P_i, P_j : [P_i, P_j]=0} \max \big\{ 2 \vert \Delta_i \vert, ~2 \vert \Delta_j \vert, ~2 \vert \Delta_{ij} \vert, ~\vert \Delta_i \vert +  \vert \Delta_j \vert +  \vert \Delta_{ij} \vert \big\},
    \end{equation}
    where 
    \begin{equation}
         \Delta_{i} = \frac{\textnormal{tr}\big[ P_i \Delta]}{d} ~ ~ {\rm and} ~ ~ \Delta_{ij} = \frac{\textnormal{tr}\big[ P_i P_j \Delta]}{d},
    \end{equation}
    where $d=4$. 
\end{lemma}

\begin{proof}
    This result arises from applying Lemma~\ref{lemma: general state discrimination} with an arbitrary two qubit stabilizer basis and then maximising over all stabilizer basis, i.e., 
    \begin{equation}
     P_{\rm suc}^{\rm stab, f}(\Delta) = \frac{1}{2} + \frac{1}{4}~\max_{P_i, P_j : [P_i, P_j]=0} \sum_{a,b \in \{0,1\} } \vert \textnormal{tr}\big[ \Pi^{ij}_{a,b} \Delta \big] \vert, \label{eq: sucess probability with stabs}
    \end{equation}
    where 
    \begin{equation}
        \Pi^{ij}_{a,b} = \frac{\mathbb{I} \otimes \mathbb{I} + (-1)^a P_i + (-1)^{b} P_j + (-1)^{a \oplus b} P_iP_j}{d}.
    \end{equation}
    It can be seen that 
    \begin{equation}
         \vert \textnormal{tr}\big[ \Pi^{ij}_{a,b} \Delta \big] \vert = \vert (-1)^a \Delta_i + (-1)^b \Delta_j + (-1)^{a \oplus b} \Delta_{ij} \vert, 
    \end{equation}
    where the factor of $d$ cancels due to the fact that $\textrm{tr}\big[ P_k P_l \big] = \delta_{kl} d$, where $\delta_{kl}$ is the Kronecker delta function. Hence,  
    \begin{align}
        \sum_{a,b \in \{0,1\} } \vert \textnormal{tr}\big[ \Pi^{ij}_{a,b} \Delta \big] \vert &= \vert \Delta_i + \Delta_j + \Delta_{ij} \vert + \vert - \Delta_i + \Delta_j - \Delta_{ij} \vert 
        + \vert \Delta_i - \Delta_j - \Delta_{ij} \vert + \vert - \Delta_i  - \Delta_j + \Delta_{ij} \vert.
    \end{align}
    By using the following equation: 
    \begin{equation}\label{eq: max equation}
        \max \{ \vert a \vert , \vert b \vert \} = \frac{\vert a + b \vert + \vert a - b \vert}{2},
    \end{equation}
    and grouping $\Delta_i$, $\Delta_j$, and $\Delta_{ij}$, this can be rewritten as 
    \begin{align} \label{eq: max equation}         \sum_{a,b \in \{0,1\} } \vert \textnormal{tr}\big[ \Pi^{ij}_{a,b} \Delta \big] \vert &= 2 \max \big\{ \vert \Delta_i \vert, \vert \Delta_j + \Delta_{ij} \vert \big\} + 2 \max \big\{ \vert \Delta_i \vert, \vert \Delta_j - \Delta_{ij} \vert \big\} \\ 
    \end{align}
    This object is therefore some function of the form
    \begin{align}
        f(a,b,c) &= \max \big\{ \vert c \vert, \vert a + b \vert \big\} + \max \big\{ \vert c \vert, \vert a - b\vert \big\} \\
        &= \max \big\{ 2 \vert c \vert, \vert c \vert + \vert a + b \vert, \vert c \vert + \vert a - b \vert, \vert a + b \vert + \vert a-b \vert  \big\}.
    \end{align}
    It can then be seen that 
    \begin{equation}
        \begin{split}
            \max\big\{  \vert a + b \vert,  \vert a - b \vert \big\} &=  \frac{\vert a+b+a-b \vert + \vert a+b-a+b \vert}{2} \\
            &=  \vert a \vert + \vert b \vert, 
        \end{split}
    \end{equation}
    It can therefore be seen that 
    \begin{equation}
        \max \big\{ \vert c \vert + \vert a + b \vert, \vert c \vert + \vert a - b \vert  \big\} = \vert a \vert + \vert b \vert + \vert c \vert.
    \end{equation}
    Hence, it can be concluded, using Eq.\eqref{eq: max equation}, that 
    \begin{equation}
         f(a,b,c) =  \max\big\{2 \vert a \vert, 2 \vert b \vert, 2 \vert c \vert,  \vert a \vert + \vert b \vert + \vert c \vert \big\}
    \end{equation}
    such that 
    \begin{equation}
        \sum_{a,b \in \{0,1\} } \vert \textnormal{tr}\big[ \Pi^{ij}_{a,b} \Delta \big] \vert = 2 \max\big\{ 2 \vert \Delta_i \vert, 2 \vert \Delta_j \vert, 2 \vert \Delta_{ij} \vert, \vert \Delta_i \vert + \vert \Delta_j \vert + \vert \Delta_{ij} \vert \big\}.
    \end{equation}
    Inputting this into Eq~\eqref{eq: sucess probability with stabs} completes the forward direction of the proof. To see the reverse, it can simply be noted that any pair of mutually commuting two qubit Pauli-strings corresponds to a stabilizer basis. Hence, given any two-qubit commuting Pauli-strings, a basis extracting them can always be found. 
\end{proof}

\hspace{-0.4cm}{\bf Single Qubit With Qubit Ancilla and Fixed Measurements}
\\
\\
We now return to the proof of the following Lemma: 
\\
\\
{\em {\bf Result~\ref{result: for qubits, not about of resource is useful}}
$P_{\rm suc}^{\rm stab}(\vec \Delta) = P_{\rm suc}^{\rm stab, f}(\vec \Delta, \omega)~\forall~\omega \in \mathcal{D}(\mathcal{H}_1)$.
}
\begin{proof}
    To prove this, Lemma~\ref{lemma: appendix : two qubit min error sucess prob} is employed with 
    \begin{align}
        \Tilde{\Delta} &= \rho \otimes \omega - \sigma \otimes \omega \\
        &= {\Delta} \otimes \omega,
    \end{align}
    where 
    \begin{equation}
        {\Delta} \coloneq \rho - \sigma = \Delta_x X + \Delta_y Y + \Delta_z Z,
    \end{equation}
    such that $\Vec{\Delta} = (\Delta_x, \Delta_y, \Delta_z)^t$ is the difference of the Bloch vectors of $\rho$ and $\sigma$, and 
    \begin{equation}
        \omega = \frac{1}{2}(\mathbb{I} + n_x X + n_y Y + n_z Z),
    \end{equation}
    such that $\Vec{n} = (n_x, n_y, n_z)^t$ is the Bloch vector of $\omega$. In terms of these Bloch vectors, $\Tilde{\Delta}$ is then 
    \begin{equation}
        \Tilde{\Delta} = \frac{1}{4}\Vec{\Delta} \cdot \Vec{\kappa} \otimes \mathbb{I} + \sum_{ij=1}^3 T_{ij} \kappa_i \otimes \kappa_j
    \end{equation}
    where 
    \begin{equation}
        \kappa_1 = X, \kappa_2=Y, \kappa_3 = Z ~ ~ {\rm and} ~ ~ \Vec{\kappa} = (\kappa_1, \kappa_2, \kappa_3)^t,
    \end{equation}
    and $T_{ij}$ are the matrix elements of the $3 \times 3$ matrix 
    \begin{equation}
        (\Vec{\Delta})(\Vec{n})^t = \frac{1}{4}~\begin{bmatrix} \label{eq: correlation matrix}
            \Delta_x n_x & \Delta_x n_y & \Delta_x n_z \\
            \Delta_y n_x & \Delta_y n_y & \Delta_y n_z \\
            \Delta_z n_x & \Delta_z n_y & \Delta_z n_z \\
        \end{bmatrix}
    \end{equation}
    To proceed, we firstly note that a simplification of the optimisation over all stabilizer basis (Pauli-strings) can be made: if the two qubit Clifford unitary used to set the stabilizer basis being measured is a tensor product of two single qubit Clifford, it can be seen that no advantage can be gained. Specifically, letting $U = U_1 \otimes U_2 \in \mathcal{C}_2$ where $U_1, U_2 \in \mathcal{C}_1$, this can be seen as 
    \begin{align}
        \sum_{a,b \in \{0,1\}} \vert \textrm{tr}\big[ \ketbra{ab}U(\Delta \otimes \omega)U^\dagger \big] \vert &= \sum_{a,b \in \{0,1\}} \vert \textrm{tr}\big[(\ketbra{a} \otimes \ketbra{b})(U_1 \otimes U_2)(\Delta \otimes \omega)(U_1 \otimes U_2)^\dagger \big] \vert \\
        &= \sum_{a,b \in \{0,1\}} \vert \textrm{tr}\big[ \ketbra{a}U_1 \Delta U_1^\dagger \big] \vert ~\vert \textrm{tr}\big[ \ketbra{b} U_2 \omega U_2^\dagger \big] \vert \\
        &= \sum_{a \in \{0,1\}} \vert \textrm{tr}\big[ \ketbra{a}U_1 \Delta U_1^\dagger \big] \vert ~\textrm{tr}\bigg[ \bigg(\sum_{b \in \{0,1\}} \ketbra{b} \bigg) U_2 \omega U_2^\dagger \bigg],
    \end{align}
    where the absolute value can be dropped on the second term as $U_2 \omega U_2^\dagger$ is a quantum state and $\ketbra{b}$ a POVM element such that \hbox{$\textrm{tr}\big[ \ketbra{b} U_2 \omega U_2^\dagger \big] \geq 0~\forall~ \ketbra{b}~~{\rm and}~ ~U_2$}. We then note that $\sum_{b \in \{0,1\}} \ketbra{b} = \mathbb{I}$, meaning that 
    \begin{equation}
         \sum_{a,b \in \{0,1\}} \vert \textrm{tr}\big[ \ketbra{ab}U(\Delta \otimes \omega)U^\dagger \big] \vert = \sum_{a \in \{0,1\}} \vert \textrm{tr}\big[ \ketbra{a}U_1 \Delta U_1^\dagger \big] \vert,
    \end{equation}
    such that if a unitary of the form $U=U_1 \otimes U_2$ is used the success probability is independent of the state in the ancillary system. Physically, this is of course an obvious statement: if one acts independently on a system and ancilla which are initially uncorrelated (in a product state), nothing done to the ancilla will impact the system. As a result of this, one does not need to consider measuring in basis of the form 
    \begin{equation}
        \bigg\{ \Pi^{ij}_{a,b} =  \frac{\mathbb{I}_A \otimes \mathbb{I}_B+(-1)^a~P_i \otimes \mathbb{I} + (-1)^b \mathbb{I} \otimes P_j + (-1)^{a \oplus b} P_i \otimes P_j}{4} : a,b \in \{ 0,1 \}, ~P_i, P_j \in \mathcal{P}_1, ~[P_i, P_j]=0 ~\bigg\},
    \end{equation}
    as such basis are generated by applying a pre-processing tensor product of single qubit Clifford unitaries. Hence, only non-product stabilizer basis that consist of summations over the following set of Pauli strings need to be considered:
    \begin{align*}
        & \{ X \otimes X, Y \otimes Y, Z \otimes Z \} \\
        & \{ X \otimes Y, Y \otimes X, Z \otimes Z \} \\
        & \{ X \otimes Z, Y \otimes Y, Z \otimes X \} \\
        & \{ X \otimes X, Y \otimes Z, Z \otimes Y \} \\
        & \{ X \otimes Y, Y \otimes Z, Z \otimes X \} \\
        & \{ X \otimes Z, Z \otimes Y, Y \otimes X \} 
    \end{align*}
    Considering this, and Lemma~\ref{lemma: appendix : two qubit min error sucess prob}, it can be seen that the maximum success probability of discrimination is related to either the maximum absolute value of a matrix element in Eq.\eqref{eq: correlation matrix}, or the largest sum of the absolute value of three matrix elements where each comes from a distinct row and column. We now assume that the optimal basis has been found, and then determine whether either of these values could give an improvement over the success probability of measuring only the system using stabilizer circuits. 

    Initially, the largest sum of the absolute value of three matrix elements, where each element comes from a distinct row and column of Eq.\eqref{eq: correlation matrix}, is investigated. It can be seen that this will be given by
    \begin{equation}
        \frac{1}{4} ~\vert \Vec{\Delta} \vert ^\downarrow \cdot \vert \Vec{n} \vert ^\downarrow, \label{eq: max over the correlation matrix}
    \end{equation}
    where $\vert \vec{a} \vert^\downarrow$ means taking the absolute value of each of the elements of the vector $\vec{a}$ and then placing the elements in non-increasing order. This is optimal as it takes the product of the largest absolute value of $\Vec{\Delta}$ with the largest absolute value of $\Vec{n}$ and then the second largest, and so on. 
    
    As this is a dot product, this value will be maximised if $\vert \Vec{n} \vert ^\downarrow$ points in the same direction as $ \vert \Vec{\Delta} \vert ^\downarrow$. Let 
    \begin{equation}
        \vert \Vec{\Delta} \vert^\downarrow = \begin{pmatrix}
            \vert \Delta_a \vert \\
            \vert \Delta_b \vert \\
            \vert \Delta_c \vert \\
        \end{pmatrix}, ~ ~ \vert \Vec{n} \vert^\downarrow = \begin{pmatrix}
            \vert n_a \vert \\
            \vert n_b \vert \\
            \vert n_c \vert 
        \end{pmatrix},
    \end{equation}
    where $a,b,c$ is some labelling to order both vectors in non-increasing order, i.e., such that
    \begin{align} \label{eq: assumptions on ordereding}
        \vert \Delta_a \vert \geq \vert \Delta_b \vert \geq \vert \Delta_c \vert \\
        \vert n_a \vert \geq \vert n_b \vert \geq \vert n_c \vert.
    \end{align}
    Eq.\eqref{eq: max over the correlation matrix} is then maximised if 
    \begin{equation}
        \vert \Vec{n} \vert^\downarrow = \frac{1}{\sqrt{\vert \Delta_a \vert^2 + \vert \Delta_b \vert^2 + \vert \Delta_c \vert ^2}} \begin{pmatrix}
            \vert \Delta_a \vert \\
            \vert \Delta_b \vert \\
            \vert \Delta_c \vert \\
        \end{pmatrix},
    \end{equation}
    such that 
    \begin{equation}
        \vert \Vec{\Delta} \vert ^\downarrow \cdot \vert \Vec{n} \vert ^\downarrow = \sqrt{\vert \Delta_a \vert^2 + \vert \Delta_b \vert^2 + \vert \Delta_c \vert ^2}.
    \end{equation}
    The success probability would therefore be 
    \begin{equation}
        P_{\rm suc}^{\rm stab, f}(\omega) = \frac{1}{2} + \frac{1}{2} \times \frac{1}{4} \sqrt{\vert \Delta_a \vert^2 + \vert \Delta_b \vert^2 + \vert \Delta_c \vert ^2}.
    \end{equation}
    For
    $ P_{\rm suc}^{\rm stab, f}(\Delta, \omega) > P_{\rm suc}^{\rm stab}(\Delta),
    $ 
    it must therefore be the case that 
    \begin{equation}
        \vert \Delta_b \vert^2 + \vert \Delta_c \vert^2 > 3 \vert \Delta_a \vert^2,
    \end{equation}
    which is never possible due to the relations in Eq.\eqref{eq: assumptions on ordereding} (they upper bound the left-hand side by $2 \vert \Delta_a \vert^2$). Hence, there exists no $\omega$ such that, even when optimising over all stabilizer circuits, the success probability can be improved upon in this way. 
    
    We now turn to the maximum absolute value of a matrix element in Eq.\eqref{eq: correlation matrix}.
    It can be seen that this is given by $\frac{1}{4}~\vert \vert \Vec{\Delta} \vert \vert_\infty  \vert \vert \Vec{n} \vert \vert_\infty$, i.e., it is the product of largest components of the absolute values of $\Vec{\Delta}$ and $\Vec{n}$. This would give a success probability of
    \begin{align}
         P_{\rm suc}^{\rm stab, f}(\vec \Delta, \omega) &= \frac{1}{2} + \frac{1}{2} \times 2 \times \frac{1}{4}~\vert \vert \Vec{\Delta} \vert \vert_\infty  \vert \vert \Vec{n} \vert \vert_\infty \\
         &=\frac{1}{2} + \frac{1}{4} \vert \vert \Vec{\Delta} \vert \vert_\infty  \vert \vert \Vec{n} \vert \vert_\infty \\
         &\leq P_{\rm suc}^{\rm stab}(\Delta)
    \end{align}
    with equality if and only if $\vert \vert \Vec{n} \vert \vert_\infty = 1$. This final condition means only if $\omega$ is a pure stabilizer state does there exists a measurement on the joint system and ancillary space that gives a success probability equal to what could be achieved by measuring the system alone. 

    In summary, this means that there exists no $\omega$ such that using fixed stabilizer circuits the success probability of discriminating $\Delta$ can be improved. 
\end{proof}

\subsection{Proof of Qubit Helstrom Bound}\label{appendix: Proof of Qubit Helstrom Bound}

\begin{theorem}[Qubit Helstrom Bound \cite{Helstrom1976}] \label{qubit: helstrom} 
Let $\rho, \sigma \in \mathcal{D}(\mathcal{H}_1)$. The optimal minimum error success probability when discriminating  between $\rho$ and $\sigma$ is given by 
\begin{equation}
    P_{\rm suc}^*(\vec \Delta) = \frac{1}{2} + \frac{1}{4} \vert \vert  \Vec{\Delta} \vert \vert_2,
\end{equation}
where $\Vec{\Delta} = n_\rho - n_\sigma$, with $n_\rho, n_\sigma$ the Bloch vectors of $\rho$ and $\sigma$ respectively, and  
\begin{equation}
    \vert \vert \Vec{v} \vert \vert_2 = \sqrt{\vert v_1 \vert^2 + \vert v_2 \vert^2 + \vert v_3 \vert^2},
\end{equation}
is the vector $l$-$2$ norm. 
\end{theorem}

\begin{proof}
    In the computational basis $\Delta = \rho - \sigma$ is given by   
    \begin{equation}
        \Delta = \frac{1}{2}\begin{bmatrix}
            \Delta_z & \Delta_x - i \Delta_y \\
            \Delta_x + i \Delta_y & - \Delta_z
        \end{bmatrix}
    \end{equation}
    where 
    \begin{equation}
        \Vec{\Delta} = 
        \begin{bmatrix}
           \Delta_x \\
           \Delta_y \\
           \Delta_z
        \end{bmatrix} =  \begin{bmatrix}
            \textrm{tr}\big[\Delta X \big] \\
            \textrm{tr}\big[\Delta Y \big] \\
            \textrm{tr}\big[\Delta Z \big] \\
        \end{bmatrix} = 
        \Vec{n}_\rho - \Vec{n}_\sigma = \begin{bmatrix}
            \textrm{tr}\big[\rho X \big] - \textrm{tr}\big[\sigma X \big] \\
            \textrm{tr}\big[\rho Y \big] - \textrm{tr}\big[\sigma Y \big]\\
            \textrm{tr}\big[\rho Z \big] - \textrm{tr}\big[\sigma Z \big] \\
        \end{bmatrix}
    \end{equation}
    The eigenvalues of $\Delta$ are then 
    \begin{align*}
        \lambda_+ = \frac{1}{2}\sqrt{\Delta_z^2 + \Delta_x^2 + \Delta_y^2} \\
        \lambda_- = -\frac{1}{2}\sqrt{\Delta_z^2 + \Delta_x^2 + \Delta_y^2}. \\
    \end{align*}
    As $\Delta$ is a Hermitian operator, $\Delta^\dagger = \Delta$, the singular values of $\Delta$ are the absolute value of the eigenvalues. Then, the one norm is the sum of the singular values and hence 
    \begin{align*}
        \onenorm{\Delta} &= \vert \lambda_+ \vert + \vert \lambda_- \vert \\
        &= 2 \vert \lambda_+ \vert \\
        &= \vert \vert \Vec{\Delta} \vert \vert_2.
    \end{align*}
    Helstrom's bound is therefore 
    \begin{equation}
        \begin{split}
            P_{\rm suc}^*(\Delta) &= \frac{1}{2} + \frac{1}{4} \vert \vert \Delta \vert \vert_1 \\
            &= \frac{1}{2} + \frac{1}{4} \vert \vert \Vec{\Delta} \vert \vert_2,
        \end{split}
    \end{equation}
    completing the proof. 
\end{proof}

\subsection{Proof of Result~\ref{result: appending any qubit state with feedfoward}} \label{appendix: appending any qubit state with feedfoward}
As above in Appendix~\ref{appendix: result: for qubits, not about of resource is useful}, we first prove a Lemma for the minimum error success probability under adaptive measurements when $\rho, \sigma \in \mathcal{D}(\mathcal{H}_2)$. In this case, the figures of merit is 
\begin{equation}
    \begin{split}
        P_{\rm suc}^{\rm stab, a}(\Delta) &= \frac{1}{2} + \frac{1}{4} \max_{U \in \mathcal{C}_2}  \max_{\{V_b : V_b \in \mathcal{C}_1 \}_b} \sum_{a \in \{0,1\}^1} \sum_{b \in \{0,1\}^1} \big\vert \textrm{tr} \big[ (U^\dagger ( V^\dagger_b \ketbra{a} V_b \otimes \ketbra{b})U) \Delta \big] \big\vert,
    \end{split}
\end{equation}
where for $P_{\rm suc}^{\rm stab, a}(\Delta)$, after a global Clifford unitary, the outcome of a computational basis measurement on the second qubit determines which local Clifford unitary is then applied to the first qubit, before another computational basis measurement is made. Again, as above, we note that for two qubits a stabilizer basis is generated by two Pauli-strings, $P_i$ and $P_j$. Once the above expressions is known, it is then applied to the input state $(\cdot) \otimes \omega$. 
\\
\\
\hspace{-0.4cm}{\bf Two Qubit Adaptive Measurements}

\begin{lemma}\label{lemma: appendix : two qubit min error sucess prob with feed-forward}
    The minimum error success probability with feed-forward when $\Delta$ is the difference of two two-qubit density operators is 
    \begin{equation} \label{result: stab min error equality WRONG?}
        P_{\rm suc}^{\rm stab, a}(\Delta) = \frac{1}{2} + \frac{1}{2} \max_{P_i, P_j, P_k : ~[P_i, P_j]=0, ~[P_k, P_j]=0} \bigg[ \max \big\{ \vert \Delta_j \vert, \vert \Delta_i \pm {\rm sign}(P_iP_j) \Delta_{ij} \vert \big\} + \max \big\{ \vert \Delta_j \vert, \vert \Delta_k \mp  {\rm sign}(P_kP_j) \Delta_{kj} \vert \big\} \bigg], \\
    \end{equation}
    where 
    \begin{equation}
         \Delta_{i} = \frac{\textnormal{tr}\big[ P_i \Delta]}{d} ~ ~ {\rm and} ~ ~ \Delta_{ij} = \frac{\textnormal{tr}\big[ P_i P_j \Delta]}{d},
    \end{equation}
    and $d=4$.
\end{lemma}
\begin{proof}
    Given an input on a two qubit Hilbert space $(\mathcal{H}_1^A \otimes \mathcal{H}_1^B)$, the action of an adaptive two qubit stabilizer circuit on a state $\rho \in (\mathcal{H}_1^A \otimes \mathcal{H}_1^B)$ can be captured by the following instrument: 
    \begin{equation}
        \mathcal{E}(\rho) = \sum_{a,b \in \{0,1\}} (\ketbra{a} \otimes \mathbb{I})(V_b \otimes \mathbb{I})(\mathbb{I} \otimes \ketbra{b})U(\rho)U^\dagger(\mathbb{I} \otimes \ketbra{b})(V_b^\dagger \otimes \mathbb{I})(\ketbra{a} \otimes \mathbb{I}) \otimes \ketbra{ab},
    \end{equation}
    where $U \in \mathcal{C}_2$ and $V_b \in \mathcal{C}_1~\forall~b$, as detailed in Appendix~\ref{appendix: stabilzer measurements, adaptive measurements}. The probability of getting the outcome bit-string $(a,b)$ for an input $\rho$, denoted $P(a,b)_\rho$, is therefore 
    \begin{align*}
        P(a,b)_\rho &= \textrm{tr}\big[  (\ketbra{a} \otimes \mathbb{I})(V_b \otimes \mathbb{I})(\mathbb{I} \otimes \ketbra{b})U(\rho)U^\dagger(\mathbb{I} \otimes \ketbra{b})(V_b^\dagger \otimes \mathbb{I})(\ketbra{a} \otimes \mathbb{I})   \big] \\
        &= \textrm{tr}\big[  (\ketbra{a} \otimes \ketbra{b})(V_b \otimes \mathbb{I})U(\rho)U^\dagger (V_b^\dagger \otimes \mathbb{I}) \big] \\
        &= \textrm{tr}\big[ U^\dagger ( V_b^\dagger \ketbra{a} V_b \otimes \ketbra{b})U \rho \big]. \\
    \end{align*}   
    To get a picture of what is happening here, we will firstly explicitly state the operators $ V_b^\dagger\ketbra{a} V_b \otimes \ketbra{b}$ for each output bit-string $(a,b)$:
    \begin{align}
        V_0^\dagger\ketbra{0} V_0 \otimes \ketbra{0} &= \frac{\mathbb{I} \otimes \mathbb{I} + V^\dagger_0 Z V_0 \otimes \mathbb{I} + \mathbb{I} \otimes Z + V_0^\dagger Z V_0 \otimes Z}{4} \\
        V_0^\dagger \ketbra{1} V_0 \otimes \ketbra{0} &= \frac{\mathbb{I} \otimes \mathbb{I} - V^\dagger_0 Z V_0 \otimes \mathbb{I} + \mathbb{I} \otimes Z - V_0^\dagger Z V_0 \otimes Z}{4} \\
        V_1^\dagger \ketbra{0} V_1 \otimes \ketbra{1} &= \frac{\mathbb{I} \otimes \mathbb{I} + V^\dagger_1 Z V_1 \otimes \mathbb{I} - \mathbb{I} \otimes Z - V_1^\dagger Z V_1 \otimes Z}{4} \\
        V_1^\dagger \ketbra{1} V_1 \otimes \ketbra{1} &= \frac{\mathbb{I} \otimes \mathbb{I} - V^\dagger_1 Z V_1 \otimes \mathbb{I} - \mathbb{I} \otimes Z + V_1^\dagger Z V_1 \otimes Z}{4}. \\
    \end{align}
    As $V_0, V_1 \in \mathcal{C}_1$ it means that $V_0^\dagger Z V_0, V_1^\dagger Z V_1 \in \mathcal{P}_1$. Hence, $\big\{ V_0^\dagger \ketbra{0} V_0 \otimes \ketbra{0}, V_0^\dagger \ketbra{1} V_0 \otimes \ketbra{0} \big\}$ are two of the four rank one projectors that form a stabilizer basis. The same is then true for $\big\{ V_0^\dagger \ketbra{0} V_0 \otimes \ketbra{1}, V_0^\dagger\ketbra{1} V_0 \otimes \ketbra{1} \big\}$. Note if $V_0 = V_1$ then together these four operators will form a complete basis. 

    Applying a global Clifford unitary to each of these operators will then map each of the Pauli-strings to another Pauli-string and will preserve the commuting strucutre:  
    \begin{equation}
        \begin{split} \label{eq: decomposition of effect}
        U^\dagger(V_0^\dagger\ketbra{0} V_0 \otimes \ketbra{0})U &= \frac{\mathbb{I} \otimes \mathbb{I} + P_i + P_j + P_iP_j}{4}  = \ketbra{\Phi_1} \\
        U^\dagger(V_0^\dagger\ketbra{1} V_0 \otimes \ketbra{0})U &= \frac{\mathbb{I} \otimes \mathbb{I} - P_i + P_j - P_iP_j}{4} = \ketbra{\Phi_2}\\
        U^\dagger(V_1^\dagger\ketbra{0} V_1 \otimes \ketbra{1})U &= \frac{\mathbb{I} \otimes \mathbb{I} + P_k - P_j - P_kP_j}{4} = \ketbra{\Psi_3} \\
        U^\dagger(V_1^\dagger\ketbra{1} V_1 \otimes \ketbra{1})U &= \frac{\mathbb{I} \otimes \mathbb{I} - P_k - P_j + P_kP_j}{4} = \ketbra{\Psi_4} \\
        \end{split}
    \end{equation}
    where 
    \begin{align*}
        \{ \ketbra{\Phi_1}, \ketbra{\Phi_2}, \ketbra{\Phi_3}, \ketbra{\Phi_4} \}  ~ ~ {\rm and} ~ ~ 
        \{ \ketbra{\Psi_1}, \ketbra{\Psi_2}, \ketbra{\Psi_3}, \ketbra{\Psi_4} \}
    \end{align*}
    are two two-qubit stabilizer basis. 
    
    It is known, and can easily be seen from the above decomposition, that $\big\{ \ketbra{\Phi_1}, \ketbra{\Phi_2}, \ketbra{\Psi_3}, \ketbra{\Psi_4} \big\} $
    is a POVM --- each element is positive and they sum to the identity. However, it is not necessarily a stabilizer POVM, as if $[P_i, P_j] \neq 0$, it cannot be generated from computational basis states and Clifford unitaries alone. 
    
    As a result of this being a POVM, Lemma~\ref{lemma: general state discrimination} can be employed for calculating the success probability:
    \begin{equation}\label{eq: p_suc in the two qubit feed forward case}
    \begin{split}
        P_{\rm suc}^{\rm stab, a}(\Delta) &= \frac{1}{2} + \frac{1}{4} \vert \textrm{tr}\big[ \ketbra{\Phi_1} \Delta \big] \vert + \vert \textrm{tr}\big[ \ketbra{\Phi_2} \Delta  \big] \vert + \vert \textrm{tr}\big[ \ketbra{\Psi_3} \Delta  \big] \vert + \vert \textrm{tr}\big[ \ketbra{\Psi_4} \Delta  \big] \vert \\
        &= \frac{1}{2} + \frac{1}{4} \bigg[  \vert \Delta_i + \Delta_j + \Delta_{ij} \vert + \vert - \Delta_i + \Delta_j -  \Delta_{ij} \vert + \vert \Delta_k - \Delta_j - \Delta_{kj} \vert + \vert - \Delta_k - \Delta_j + \Delta_{kj} \vert \bigg] \\ 
        &= \frac{1}{2} + \frac{1}{4} \bigg[  \vert\Delta_j + (\Delta_i + \Delta_{ij}) \vert + \vert \Delta_j - (\Delta_i +  \Delta_{ij}) \vert + \vert -\Delta_j + (\Delta_k - \Delta_{kj}) \vert + \vert - \Delta_j - (\Delta_k - \Delta_{kj}) \vert \bigg] \\ 
        &= \frac{1}{2} + \frac{1}{4} \bigg[ 2 \max\big\{ \vert \Delta_j \vert, \vert \Delta_i + \Delta_{ij} \vert \big\} +  2 \max\big\{ \vert \Delta_j \vert, \vert \Delta_k - \Delta_{kj} \vert \big\}\bigg].
    \end{split}
    \end{equation}
    We now focus on determining when the second term in each of the maximums is the sum of two Pauli-coefficients of $\Delta$, and when it is the difference of two Pauli-coefficients of $\Delta$. What determines if this term is a difference or a sum is whether the sign of $P_i$ and $P_iP_j$ are the same, or different. In Eq.~\eqref{eq: decomposition of effect} it can be seen that this sign structure is initially set by which output bit-string one has. For the outcome $b=0$ the signs are the same, such that the term $\vert \Delta_i + \Delta_{ij} \vert$ appears in $P_{\rm suc}^{\rm stab, a}(\Delta)$. For $b=1$, the signs are different and hence the term $\vert \Delta_k - \Delta_{kj} \vert$ appears in $P_{\rm suc}$. Therefore, for a given $P_i, P_j, P_k$ the positive and negative signs are interchangeable by instead applying $V_1$ if $b=0$ and $V_0$ if $b=1$, such that:
    \begin{equation}
       P_{\rm suc}^{\rm stab, a}(\Delta) = \frac{1}{2} + \frac{1}{4} \bigg[ 2 \max\big\{ \vert \Delta_j \vert, \vert \Delta_i \pm \Delta_{ij} \vert \big\} +  2 \max\big\{ \vert \Delta_j \vert, \vert \Delta_k \mp \Delta_{kj} \vert \big\}\bigg].
    \end{equation}
    We then note that if the product $P_iP_j$ has an overall phase of $-1$, then the sign of the extracted coefficient $\Delta_{ij}$ will be flipped --- this is likewise the case for $P_kP_j$. Hence, $P^{\rm stab, a}_{\rm suc}(\Delta)$ is 
    \begin{align*}
    P_{\rm suc}^{\rm stab, a}(\Delta) = \frac{1}{2} + \frac{1}{4} \bigg[ 2 \max\big\{ \vert \Delta_j \vert, \vert \Delta_i \pm {\rm sign}(P_iP_j) \Delta_{ij} \vert \big\} +  2 \max\big\{ \vert \Delta_j \vert, \vert \Delta_k  \mp {\rm sign}(P_kP_j)  \Delta_{kj} \vert \big\}\bigg],
    \end{align*}
    where ${\rm sign(P_iP_j)} \in \{-1, +1\}$ is the overall phase of the product of the Pauli-strings $P_i$ and $P_j$. Therefore, the different sign configurations that can be created are limited by the structure of the available stabilizer basis, with arbitrary sign configurations not being able to be created i.e., one cannot just choose $+$ or $-$ in each term. 
    To complete the proof, we must then show that whenever there exists two qubit Pauli-strings $P_i, P_j, P_k \in \mathcal{P}_2$ such that 
    \begin{equation}
        [P_i, P_j] = 0 ~ ~ {\rm and} ~ ~ [P_k, P_j]=0,
    \end{equation}
    it is possible to find a two qubit Clifford unitary $U \in \mathcal{C}_2$, and two single qubit Clifford unitaries $V_0, V_1 \in \mathcal{C}_1$, such that 
    \begin{align*}
        P_i &= U^\dagger (V_0^\dagger \otimes \mathbb{I})({Z} \otimes \mathbb{I})(V_0 \otimes \mathbb{I})U \\
        P_j &= U^\dagger (\mathbb{I} \otimes Z) U \\
        P_k & = U^\dagger (V_1^\dagger \otimes \mathbb{I})({Z} \otimes \mathbb{I})(V_1 \otimes \mathbb{I})U.
    \end{align*}
    To see this, we first note that there exists a Clifford unitary that maps each Pauli-string to every other Pauli-string, hence there always exists a $U^\dagger \in \mathcal{C}_2$ such that 
    \begin{equation}
        P_j = U^\dagger(\mathbb{I} \otimes Z)U. 
    \end{equation}
    Considering now the Pauli-strings $P'_i=U P_i U^\dagger$, such that $P_i = U^\dagger P_i'U$. As commutation relations are preserved under conjugation by Clifford unitaries it follows that 
    \begin{equation}
        [P_i, P_j] = 0 \iff [P'_i, \mathbb{I} \otimes Z] = 0. 
    \end{equation}
    This then implies that 
    \begin{equation}
        P'_i \in \{A \otimes \mathbb{I}, A \otimes Z\},
    \end{equation}
    for $A \in \mathcal{P}_1$. As there exists single qubit Clifford unitaries that map each Pauli to every other Pauli, there always exists a $V_0^\dagger \in \mathcal{C}_1$ such that $A=V_0^\dagger Z V_0$, hence 
    \begin{equation} \label{eq: different second Pauli-terms}
        P_i' = (V_0^\dagger \otimes \mathbb{I})(Z \otimes \mathbb{I})(V_0 \otimes \mathbb{I}) ~ ~ {\rm or} ~ ~ P_i' = (V_0^\dagger \otimes \mathbb{I})(Z \otimes {Z})(V_0 \otimes \mathbb{I}).
    \end{equation}
    One can then conjugate $P_i'$ via $U^\dagger$ to get $P_i$. For the first case in Eq.\eqref{eq: different second Pauli-terms}, the proof is complete for $P_i$. For the latter case, it can be seen that 
    \begin{align*}
        P_iP_j &= \big[U^\dagger(V_0^\dagger \otimes \mathbb{I})(Z \otimes {Z})(V_0 \otimes \mathbb{I})U \big] ~\big[ U^\dagger(\mathbb{I} \otimes Z) U \big]\\
        &= U^\dagger (V_0^\dagger \otimes \mathbb{I})(Z \otimes \mathbb{I})(V_0 \otimes \mathbb{I})U,
    \end{align*}
    which means that under a relabelling $P_i \rightarrow P_iP_j$ and $P_iP_j \rightarrow P_i$ the conditions are recovered. This does not effect the success probability as the resulting sign permutation has no effect due to the absolute value. The above proof will also hold for $P_j$. This completes the proof. 
\end{proof}

\hspace{-0.4cm}{\bf Single Qubit With Qubit Ancilla and Adaptive Measurements}
\\
\\
Lemma~\ref{lemma: appendix : two qubit min error sucess prob with feed-forward} says that when trying to discriminate between two two-qubit states with feed-forward the sucess probability takes one of four values, assuming the optimal $P_i, P_j, P_k$ have been found: 
\begin{enumerate}
    \item $P_{\rm suc}^{\rm stab, a}(\Delta) = \frac{1}{2} + \vert \Delta_j \vert$ 
    \item $P_{\rm suc}^{\rm stab, a}(\Delta) = \frac{1}{2} + \frac{1}{2}\big[ \vert \Delta_i \pm {\rm sign}(P_iP_j) \Delta_{ij} \vert + \vert \Delta_j \vert \big] \leq \frac{1}{2} + \frac{1}{2} \big[ \vert \Delta_i \vert + \vert \Delta_j \vert + \vert \Delta_{ij} \vert \big]. $
    \item $P_{\rm suc}^{\rm stab, a}(\Delta) = \frac{1}{2} + \frac{1}{2}\big[ \vert \Delta_j \vert + \vert \Delta_k \mp  {\rm sign}(P_kP_j) \Delta_{kj} \vert \big] \leq \frac{1}{2} + \frac{1}{2} \big[ \vert \Delta_k \vert + \vert \Delta_j \vert + \vert \Delta_{kj} \vert \big].  $
    \item $P_{\rm suc}^{\rm stab, a}(\Delta) = \frac{1}{2} + \frac{1}{2} \big[ \vert \Delta_i  \pm {\rm sign}(P_iP_j) \Delta_{ij} \vert + \vert \Delta_k  \mp {\rm sign}(P_kP_j) \Delta_{kj} \vert \big] $.
\end{enumerate}
It can be noted that $1,2$ and $3$ can all be achieved without feed-forward, as they all consider measuring the coefficients of Pauli-strings that belong to a single stabilizer basis (See Lemma~\ref{lemma: appendix : two qubit min error sucess prob}). The $4th$ possibility, however, includes terms from two different stabilizer basis, and therefore has the possibility to give an advantage over the case without feed-forward. This observation will be used in the next proof. 
\\
\\
Let $\rho, \sigma \in \mathcal{D}(\mathcal{H}_1)$ and $\Vec{\Delta} \coloneq n_\rho - n_\sigma$ where $n_\rho$ and $n_\sigma$ are the Bloch vectors of $\rho$ and $\sigma$ respectively. 
\\
\\
{\em {\bf Result~\ref{result: appending any qubit state with feedfoward}}
    Let $n_\omega$ be the Bloch vector of $\omega \in \mathcal{D}(\mathcal{H}_1)$ and $\vert \vert \Vec{v} \vert \vert_{2-{\rm ky}}$ the $2$nd vector KY-Fan norm i.e., the sum of largest two components of the absolute value of $\Vec{v}$, then 
    \begin{equation*}
        P_{\rm suc}^{\rm stab,~ a}(\Vec{\Delta}, \omega) =  \frac{1}{2} + \frac{1}{4} \max \big\{\vert \vert \Vec{\Delta} \vert \vert_\infty, \vert \vert \Vec{\Delta} \vert \vert_{2-\rm{ky}} \vert \vert \Vec{n}_\omega \vert \vert_{2-\rm{ky}}{/2} \big\}
    \end{equation*}
}
\begin{proof}
    Firstly, we write $\Delta$ as 
    \begin{align}
        \Tilde{\Delta} &= \rho \otimes \omega - \sigma \otimes \omega \\
        &= {\Delta} \otimes \omega, \\
        &= \sum_{P_j \in \mathcal{P}_2} \Tilde{\Delta}_j P_j.
    \end{align}
    More specifically, 
    \begin{equation}\label{eq: more specific pauli-decompostion}
        \Tilde{\Delta} = \frac{1}{4}\Vec{\Delta} \cdot \Vec{\kappa} \otimes \mathbb{I} + \sum_{ij=1}^3 T_{ij} \kappa_i \otimes \kappa_j
    \end{equation}
    with 
    \begin{equation}
        \kappa_1 = X, \kappa_2=Y, \kappa_3 = Z ~ ~ {\rm and} ~ ~ \Vec{\kappa} = (\kappa_1, \kappa_2, \kappa_3)^t,
    \end{equation}
    and $T_{ij}$ being the matrix elements of the $3 \times 3$ matrix 
    \begin{equation}
        (\Vec{\Delta})(\Vec{n})^t = \frac{1}{4}~\begin{bmatrix} \label{eq: correlation matrix take 2}
            \Delta_x n_x & \Delta_x n_y & \Delta_x n_z \\
            \Delta_y n_x & \Delta_y n_y & \Delta_y n_z \\
            \Delta_z n_x & \Delta_z n_y & \Delta_z n_z \\
        \end{bmatrix}
    \end{equation}
    Then, we note that a simplification of the optimisation over all stabilizer basis (Pauli-strings) can be made: if the two qubit Clifford unitary used to set the stabilizer basis being measured is a tensor product of two single qubit Cliffords, it can be seen that no advantage can be gained (as considered in the proof of Result~\ref{result: for qubits, not about of resource is useful} given above). Explicitly letting $U = U_A \otimes U_B \in \mathcal{C}_2$ where $U_A, U_B \in \mathcal{C}_1$, it can be seen that: 
    \begin{equation}
        \begin{split}
             P_{\rm suc}^{\rm stab, a}(\Delta, \omega) &= \frac{1}{2} + \frac{1}{4} \max_{U \in \mathcal{C}_2}  \max_{\{V_b : V_b \in \mathcal{C}_1 \}_b} \sum_{a \in \{0,1\}^1} \sum_{b \in \{0,1\}^1} \big\vert \textrm{tr} \big[ (U^\dagger ( V^\dagger_b \ketbra{a} V_b \otimes \ketbra{b})U) (\Tilde{\Delta}) \big] \big\vert, \\
             &= \frac{1}{2} + \frac{1}{4} \sum_{a \in \{0,1\}^1} \sum_{b \in \{0,1\}^1} \vert \textrm{tr}\big[  (U_A \otimes U_B)^\dagger(V_b^\dagger \ketbra{a} V_b \otimes \ketbra{b})  (U_A \otimes U_B) (\Delta \otimes \omega) \big] \vert \\
              &= \frac{1}{2} + \frac{1}{4} \sum_{a,b \in \{0,1\}} \vert \textrm{tr}\big[ (U_A^\dagger V_b^\dagger \ketbra{a} V_b U_A \otimes U_B^\dagger \ketbra{b} U_B)(\Delta \otimes \omega) \big] \vert \\
              &= \frac{1}{2} + \frac{1}{4} \sum_{a,b \in \{0,1\}} \vert \textrm{tr}\big[ (U_A^\dagger V_b^\dagger \ketbra{a} V_b U_A (\Delta) \big] \vert ~ \vert \textrm{tr} \big[ U_B^\dagger \ketbra{b} U_B (\omega) \big] \vert \\
              &= \frac{1}{2} + \frac{1}{4} \bigg[ p \sum_{a \in \{0,1\}} \vert \textrm{tr}\big[ \Tilde{U}_{A,0}^\dagger \ketbra{a} \Tilde{U}_{A,0} (\Delta) \big] \vert + (1-p) \sum_{a \in \{0,1\}} \vert \textrm{tr}\big[ \Tilde{U}_{A,1}^\dagger \ketbra{a} \Tilde{U}_{A,1} \Delta \big] \vert \biggl] .
        \end{split}
    \end{equation}
    where $p= ~ \textrm{tr} \big[ U_B^\dagger \ketbra{0} U_B (\omega) \big] $, and $\Tilde{U}_{A,0} = V_0 U_A$ and $ \Tilde{U}_{A,1} = V_1U_A$ are Clifford unitaries, as products of Clifford unitaries are Clifford unitaries. As $p \geq 0$, the success probability can therefore be seen to be a probabilistic mixture of performing two different single qubit measurements to try and discrimination between $\rho$ and $\sigma$. One can therefore not perform better than just measuring the qubit in the optimal basis with probability one. 
    
    Hence, in order to provide an advantage when performing QSD with adaptive measurements, one needs to apply an entangling two qubit Clifford unitary. One can therefore just optimise over the measurement basis characterised by the following Pauli's:  
    \begin{equation}
        \begin{split}     \label{eq: entangled two qubit stab basis}
        & \{ X \otimes X, Y \otimes Y, Z \otimes Z \} \\
        & \{ X \otimes Y, Y \otimes X, Z \otimes Z \} \\
        & \{ X \otimes Z, Y \otimes Y, Z \otimes X \} \\
        & \{ X \otimes X, Y \otimes Z, Z \otimes Y \} \\
        & \{ X \otimes Y, Y \otimes Z, Z \otimes X \} \\
        & \{ X \otimes Z, Z \otimes Y, Y \otimes X \}.
        \end{split}
    \end{equation}
    Only terms in the Pauli-decomposition of $\Tilde{\Delta}$
    in Eq.~\eqref{eq: correlation matrix take 2} will therefore be relevant for gaining an advantage with adaptive measurements and an ancillary system. 
    
    It can be seen from Lemma~\ref{lemma: appendix : two qubit min error sucess prob with feed-forward} that the two basis being optimised over to find $P_{\rm suc}^{\rm stab}(\Delta, \omega)$ must have a single Pauli-string in common. By grouping together the sets of Pauli's in Eq.~\eqref{eq: entangled two qubit stab basis} that have a common Pauli string it can be verified that 
    \begin{equation}
        {\rm sign}(P_iP_j) =  - {\rm sign}(P_kP_j)
    \end{equation}
    for all possible pairings i.e., the products of the other Pauli-strings in the sets with a mutual Pauli-string have opposite signs. Hence, when consider just these entangled basis,  Lemma~\ref{lemma: appendix : two qubit min error sucess prob with feed-forward} becomes 
    \begin{equation} \label{result: stab min error equality WRONG?}
        P_{\rm suc}^{\rm stab, a}(\Delta, \omega) = \frac{1}{2} + \frac{1}{2} \max_{P_i, P_j, P_k : ~[P_i, P_j]=0, ~[P_k, P_j]=0} \big[ \max \big\{ \vert \Tilde{\Delta}_j \vert, \vert \Tilde{\Delta}_i \pm \Tilde{\Delta}_{ij} \vert \big\} + \max \big\{ \vert \Tilde{\Delta}_j \vert, \vert \Tilde{\Delta}_k \pm  \Tilde{\Delta}_{kj} \vert \big\} \big], \\
    \end{equation}
    where, as discussed above, we may focus on only 
    \begin{equation}
          P_{\rm suc}^{\rm stab, a}(\Delta, \omega) = \frac{1}{2} + \frac{1}{2}  \max_{P_i, P_j, P_k : ~[P_i, P_j]=0, ~[P_k, P_j]=0} \big[ \vert \Tilde{\Delta}_i \pm \Tilde{\Delta}_{ij} \vert + \vert \Tilde{\Delta}_k \pm \Tilde{\Delta}_{kj} \vert \big],
    \end{equation}
    as this is where an advantage can be found in the adaptive case over the non-adaptive case. Both terms are therefore either the sum of Pauli coefficients of $\Tilde{\Delta}$, or the difference of Pauli coefficients of $\Tilde{\Delta}$.
    Hence forth we focus on the expression:
    \begin{equation}
        \vert \Tilde{\Delta}_i \pm \Tilde{\Delta}_{ij} \vert + \vert \Tilde{\Delta}_k \pm \Tilde{\Delta}_{kj} \vert. \label{eq: key values in p_suc}
    \end{equation}
    Each $\Tilde{\Delta}$ term is the product of a term from $\Vec{\Delta}$ and $\Vec{n}$, as seen in Eq.~\eqref{eq: more specific pauli-decompostion}, i.e., for a Pauli-string $P_j = P_{j(1)} \otimes P_{j(2)}$ one has an associated coefficient $\Tilde{\Delta}_j = \frac{1}{4} \Delta_{j(1)} n_{j(2)}$. 

    Let the overlapping Pauli-string be $P_j = P_{j(1)} \otimes P_{j(2)}$. It can be seen by inspection that the coefficients of $\Delta$ from Eq.~\eqref{eq: correlation matrix take 2} that feature in Eq.~\eqref{eq: key values in p_suc} arise from crossing out the row associated to $P_{j(1)}$, and the column associated to $P_{j(2)}$. For example, consider the following to basis: 
    \begin{equation}
        \big\{ X \otimes Z, Z \otimes Y, Y \otimes X \big\} ~ ~ {\rm and} ~ ~ \big\{ X \otimes Z, Y \otimes Z, Z \otimes X \big\},
    \end{equation}
    such that the overlapping Pauli-string in the two basis is $X \otimes Z$. Then, the coefficients of Eq.~\eqref{eq: correlation matrix take 2} that are relevant are:
    \NiceMatrixOptions{cell-space-limits = 1pt}
    \begin{equation}
    \frac{1}{4}\begin{bmatrix}
    \Delta_x n_x & \Delta_x n_y & \Delta_x n_z \\
    \drawCustomCircletwo{0}{0}{0.3cm}{green}{green}{0.4}{0pt}{$\Delta_y n_x$} & \drawCustomCircletwo{0}{0}{0.3cm}{blue}{blue}{0.4}{0pt}{$\Delta_y n_y$} & \Delta_y n_z \\
    \drawCustomCircletwo{0}{0}{0.3cm}{blue}{blue}{0.4}{0pt}{$\Delta_z n_x$} & \drawCustomCircletwo{0}{0}{0.3cm}{green}{green}{0.4}{0pt}{$\Delta_z n_y$} & \Delta_z n_z
    \end{bmatrix}
    \end{equation}
    such that in this example
    \begin{equation}
        \vert \Delta_i \pm \Delta_{ij} \vert + \vert \Delta_k \pm \Delta_{kj} \vert = \frac{1}{4} \vert \Delta_y n_x \pm \Delta_z n_y \vert + \frac{1}{4} \vert \Delta_y n_y \pm \Delta_z n_x \vert.
    \end{equation}
    To find the maximum, we note that the largest four matrix elements of Eq.~\eqref{eq: correlation matrix take 2} --- ignoring the sign for now --- that can be simultaneous selected by choosing two stabilizer basis from Eq.~\eqref{eq: entangled two qubit stab basis} with an overlapping Pauli-string is:
    \begin{equation}
        \frac{1}{4}\Delta_a n_\alpha, ~ \frac{1}{4}\Delta_a n_\beta, ~ \frac{1}{4}\Delta_b n_\alpha, ~ \frac{1}{4}\Delta_b n_\beta,
    \end{equation}
    where 
    \begin{equation}
        \vert \Vec{\Delta} \vert^\downarrow = \begin{pmatrix}
            \vert \Delta_a \vert \\
            \vert \Delta_b \vert \\
            \vert \Delta_c \vert 
        \end{pmatrix} ~ ~ \vert \Vec{n} \vert^\downarrow = \begin{pmatrix}
            \vert n_\alpha \vert \\
            \vert n_\beta \vert \\
            \vert n_\gamma \vert 
        \end{pmatrix}, 
    \end{equation}
    are the absolute value of the vectors $\Vec{\Delta}$ and $\Vec{n}$ ordered in non-increasing order. If one chooses the shared mutually commuting Pauli-string to be that associated to $c$ on the first qubit and $\gamma$ on the second, then 
    \begin{equation}
        \begin{split}
            \vert \Delta_i \pm \Delta_{ij} \vert + \vert \Delta_k \pm \Delta_{kj} \vert &= \frac{1}{4}~\vert \Delta_a n_\alpha \pm \Delta_b n_\beta \vert + \frac{1}{4}~\vert \Delta_a n_\beta \pm \Delta_b n_\alpha \vert \\
            &\leq \frac{1}{4}~ \vert \Delta_a n_\alpha \vert + \vert \Delta_b n_\beta \vert + \frac{1}{4}~ \vert \Delta_a n_\beta \vert + \vert \Delta_b n_\alpha \vert, \\
            &= \frac{1}{4}~ \vert \Delta_a \vert (\vert n_\alpha \vert + \vert n_\beta \vert) + \frac{1}{4}~\vert \Delta_b \vert ( \vert n_\alpha \vert + \vert n_\beta \vert) \\
            &= \frac{1}{4}~ (\vert \Delta_a \vert + \vert \Delta_b \vert) (\vert n_\alpha \vert + \vert n_\beta \vert) \\
            &= \frac{1}{4}~ \vert \vert \Vec{\Delta} \vert \vert_{2-{\rm ky}} ~  \vert \vert \Vec{n} \vert \vert_{2-{\rm ky}}
        \end{split}
    \end{equation}
    where the upper-bound comes from the triangle inequality. 

    It must then be shown that this upper bound is always achievable. This can be seen to arise from the freedom to choose whether the terms in Eq.\eqref{eq: key values in p_suc} are the sum of two Pauli coefficients, or the difference of two Pauli coefficients (by alternating for which $b$ the unitaries $V_{0}$ and $V_1$ are applied). To do this, one can do an exhaustive consideration of all possible sign combinations of $\Delta_a, \Delta_b, n_a, n_b$ being positive or negative. In doing so, one can see that regardless of the configurations of the signs, the sum of the absolute values of the individual elements can be achieved in the minimum error success probability by using different Clifford unitaries $U, V_0$ and $V_1$ to set the basis. 
    
    This therefore means that the largest success probability that can be achieved when performing QSD with adaptive stabilizer measurements and a state $\omega$ in the ancilla is given by 
    \begin{equation}
        P_{\rm suc}^{\rm stab, a}(\Delta, \omega) = \frac{1}{2} + \frac{1}{2} \times \frac{1}{4} \vert \vert \Vec{\Delta} \vert \vert_{2-{\rm ky}} ~  \vert \vert \Vec{n} \vert \vert_{2-{\rm ky}}.
    \end{equation}
    It can easily be seen that there exists examples where this is less then $P^{\rm stab}_{\rm suc} (\Delta) = \frac{1}{2} + \frac{1}{4}\vert \vert \Vec{\Delta} \vert \vert_\infty$. This is due to the restriction to only considering the measurement basis in Eq.~\eqref{eq: entangled two qubit stab basis} i.e., the entangled stabilizer basis. However, it will sometimes be optimal to not interact with the state in the ancillary system, and instead just perform a stabilizer circuit on the unknown states. The maximum achievable minium error success probability will occur for whichever strategy gives a higher success probability. Taking the maximum of the two therefore completes the proof.  
\end{proof}

\subsection{Proof of Result~\ref{result: QSD for qubits with non-stabilizer measurements}}\label{appendix: QSD for qubits with magic measurements}

\textit{\textbf{Result}.~\ref{result: QSD for qubits with non-stabilizer measurements}~For measurements $\{M_0^\mu,M_1^\mu\}$ bounded by $\mathcal{M}(M^\mu_i) \le \mu~\forall~i$, the success probability is tightly bounded below by:
\begin{equation}
P_{\mathrm{suc}}^\mu \geq \max_{\mathcal A}P^\mathcal{A}_{\mathrm{suc}}(\mu),
\end{equation}
where {$P^\mathcal{A}_{\mathrm{suc}}(\mu)$ denotes the optimal success probability assuming that the closest stabilizer state has active support on the subspace $\mathcal{A}$}:
\begin{equation}\label{eq appendix: Psucc_Nostab_m_qbits_mt}
P^\mathcal{A}_{\mathrm{suc}}(\mu) = 
\begin{cases}
\frac{1}{2} + \frac{1}{4N}(S+C_N\vert \vert\vec\Delta\vert \vert_{\mathcal{A},1}), & \text{if } C_N < \frac{\vert \vert\vec\Delta\vert \vert_{\mathcal{A},1}}{\vert \vert\vec\Delta\vert \vert_{\mathcal{A},2}} \\
\frac{1}{2} + \frac{1}{4}\vert \vert\vec{\Delta}\vert \vert_{\mathcal{A},2}, & \text{if } C_N \ge \frac{\vert \vert\vec\Delta\vert \vert_{\mathcal{A},1}}{\vert \vert\vec\Delta\vert \vert_{\mathcal{A},2}}
\end{cases},
\end{equation}
with $S=\sqrt{(N\vert \vert\vec\Delta\vert \vert_{\mathcal{A},2}^2 - \vert \vert\vec\Delta\vert \vert_{\mathcal{A},1}^2)(N-C_N^2)}$. Here, $\vert \vert\cdot\vert \vert_{\mathcal{A},p}$ denotes the $p$-norm restricted to $\mathcal{A}$, and $N = \text{dim}(\mathcal A)$. $C_N = 1 +2\sqrt{N}\mu$  defines the surface with non-stabilizerness $\mu$ restricted to $\mathcal{A}$.}
\begin{proof}
We begin by establishing relevant quantities. Denote the Bloch vector of  the closest stabilizer state to the optimal measurement (following the measure $\mathcal{M}$) by $\vec{s}$. Then, $P^\mathcal{A}_{\mathrm{suc}}(\mu)$ corresponds to the optimal success probability restricting the measurements to the subspace $\mathcal A$ spanned by the coordinates for which $s_i>0$. $N = \text{dim}(\mathcal A)\in\{1,2,3\}$ is the number of dimensions where  $\vec{s}$ has $s_i>0$ and $C_N$ defines the flat surface corresponding to the Bloch vectors yielding measurements with magic $\mu$, $C_N=\sum_{i=1}^Nm_i = 1 +2\sqrt{N}\mu$. We show that for some of the states considered in the main text, the bound is tight, this is $P_{\mathrm{suc}}^\mu = \max_{\mathcal A}P^\mathcal{A}_{\mathrm{suc}}(\mu)$.

Now consider the success probability of successfully discriminating the states 
$\rho_0$ and $\rho_1$ with measurements $M^\mu_0,M^\mu_1$:
\begin{equation}
    P^\mu_{\text{suc}} = \frac{1}{2} + \frac{1}{2}\text{Tr}(\Delta M^\mu_0),
\end{equation}
where $\Delta = \rho_0-\rho_1$ and $M^\mu_0+M^\mu_1 = \mathbb{I}$.

Consider also the states and measurements in the Bloch vector decomposition $\rho_0 = \frac{1}{2}(\mathbb{I} + \vec{r}_0\cdot{\vec\sigma})$, $\rho_1 = \frac{1}{2}(\mathbb{I} + \vec{r}_1\cdot{\vec\sigma})$, and  $M^\mu_0=\frac{1}{2}(\mathbb{I}+\vec{m}\cdot{\vec\sigma})$. 
By denoting $\vec{\Delta} = \vec{r}_0-\vec{r}_1$, we have 
\begin{equation}
    P^\mu_{\text{suc}} = \frac{1}{2} + \frac{1}{4}\vec{\Delta} \cdot \vec{m}.
\end{equation}
Without loss of generality, we relabel the stabilizer states (or equivalently, we apply local Clifford rotations) in such a way that $\vec{\Delta}$ is in the positive octant, sorted as $\Delta_1 \geq \Delta_2 \geq\Delta_3 \geq 0$. 

We consider the non-stabilizer monotone for states  $\mathcal{M}(\rho) = \min_{\sigma \in \mathrm{STAB}} \frac{1}{2} \| \rho - \sigma \|_1$ defined in ~\cite{Cao_2025, palhares_2026}. As in Sec.~\ref{sm: SDP for non-stabilizer measurements}, we extend the definition to measurements. However, the assumed Bloch form implies $\text{Tr}(M^\mu_i)=1$ and therefore, 
\begin{equation}\label{eq:app:Trace distance}
    \mathcal{M}\left(\frac{M^\mu_i}{\text{Tr}(M^\mu_i)}\right) = \mathcal{M}\left(M^\mu_i\right)  = \min_{\sigma_i \in \mathrm{STAB}} \frac{1}{2} \|M^\mu_i- \sigma_i \|_1\leq\mu, \quad\quad for\quad i=0,1.
\end{equation}
To impose the non-stabilizer constraint, we use the closest stabilizer Bloch vector $\vec{s}$. The optimization is then performed over both $\vec m$ and $\vec {s}$, where $\vec {s}$ acts as an auxiliary variable enforcing that the measurement lies at most a distance $2\mu$ from the stabilizer polytope. Therefore the optimization can be stated as 
\begin{align}
    &\max_{\vec{m},\vec{s}} \frac{1}{2} + \frac{1}{4}\vec{\Delta} \cdot \vec{m}\\
    &s.t.\nonumber\\
    &\quad \|\vec m\|^2_2\leq 1,\nonumber\\
    &\quad \|\vec s\|_1\leq 1,\quad  \vec{s}_i\geq 0,\nonumber\\
    &\quad \|\vec  m-\vec{s}\|^2_2\leq 4\mu^2,\nonumber
\end{align}
Where the first constraint restricts the measurement vector to be inside the Bloch sphere, the second constraint restricts the closest stabilizer Bloch vector to be inside the positive quadrant of the octahedron. Finally the third constraint imposes that the trace distance between the stabilizer polytope and the measurement be less than $\mu$. We used the fact that the trace distance $\frac{1}{2} \|M^\mu_i- \sigma_i \|_1$ is half the euclidean distance $\|\cdot\|_2$ between the Bloch vectors of the closest stabilizer state $\vec{s}$ and the measurement Bloch vector $\vec m$.

To obtain a close analytical expression, we introduce the \textit{no-leakage Ansatz}: we assume that the measurement vector $\vec m$ is in the same geometric space than the stabilizer vector $\vec s$. 
This is, the no-leakage Ansatz says that if $\vec s\in\mathcal A$ then,   $\vec m\in\mathcal A$. Since the SDP presented in Sec.~\ref{sm: SDP for non-stabilizer measurements} will in general consider measurment vectors outside of $\mathcal A$, our derivation constitutes a lower bound to the unconstrained SDP optimization. However, this assumption allows us to derive a close analytical expression as follows. 

Define the norms $\|\cdot\|_{\mathcal A, 2}$ and $\|\cdot\|_{\mathcal A, 1}$ as the 2-norm and 1-norm restricted to the $\mathcal A$ space. The no-leakage ansatz then implies that $\|\vec  m-\vec{s}\|_2 = \|\vec  m-\vec{s}\|_{\mathcal A, 2}$, $\|\vec  m\|_2 = \|\vec  m\|_{\mathcal A,2}$ and $\|\vec{s}\|_1 = \|\vec{s}\|_{\mathcal A,1} $. Therefore we can restrict our optimization to the $\mathcal A $ space.  To see this formally, let us define the Lagrangian $\mathcal L$ that optimizes $P^\mu_{\text{suc}}$ subject to the $3$ constraints above. Since the non-constant part of the objective function is $\vec \Delta \cdot \vec m$, it reads
\begin{equation}
    \mathcal L(\vec m, \vec s, \nu, \lambda, \gamma) = \vec\Delta \cdot \vec m - \nu (\|\vec m\|_2^2 - 1) - \lambda (\|\vec m - \vec s\|_2^2 - 4\mu^2) - \gamma (\| \vec s\|_1 - 1)+\sum_{i=1}^3\eta_is_i,
\end{equation}
where $\nu, \lambda$, $\gamma$ and $\eta_i$ for  $i\in\{1,2,3\}$ are Lagrangian multipliers.

To obtain the maximum, we search for the points where the gradient is zero:
\begin{equation}
     \frac{\partial \mathcal L }{\partial m_i} = \Delta_i - 2\nu m_i -2\lambda(m_i-s_i) = 0,\label{eq:m_grad}
\end{equation}
and
\begin{equation}
    \frac{\partial \mathcal L }{\partial s_i} = 2\lambda(m_i-s_i)-\gamma+\eta_i = 0. \label{eq:s_grad}
\end{equation}
Note that in the last equation, since we are restricting $s_i\ge0$, we could remove the absolute values in the sum defining $\| \vec s\|_1 $. From complementary slackness, we have $\eta_is_i=0$, defining the active coordinates by the condition $s_i> 0$, $i.e$ the coordinates defining $\mathcal A$.  Therefore, for $i\in \mathcal A$, $\eta_i=0$ and Eq.~(\ref{eq:s_grad}) reads
\begin{equation}
    2\lambda(m_i-s_i)=\gamma.\label{eq:s_grad_active}
\end{equation}
Now, denote the vector $\mathbb{1}_N=(1,\dots,1)$ in $\mathbb R^N$ with $N=\text{dim}(\mathcal A)$. From the equation of the active components~(\ref{eq:s_grad_active}) and under our no-leakage ansatz, we have 

\begin{equation}
    \|\vec m-\vec s\|^2_{\mathcal{A},2} = \frac{\gamma^2\|\mathbb{1}_N\|^2_2}{4\lambda^2} =4\mu^2,
\end{equation}
which implies $\gamma/\lambda =4\mu/\|\mathbb{1}_N\|_2 = 4\mu/\sqrt{N}$. We emphasize that in general $\|\vec m-\vec s\|_2^2  = \|\vec m-\vec s\| ^2_{\mathcal{A},2} + L = 2\mu$, where $L$ depends on the other possible dimensions not spanning $\mathcal A$ (dimension leakage). The no-leakage assumption precisely takes $L=0$.  Therefore we can rewrite Eq.~\eqref{eq:s_grad_active} as $\vec m=\vec s + 2\mu\mathbb{1}_N/\sqrt{N}$ which summing over the active components we obtain the flat surface corresponding to the Bloch vectors yielding measurements with non-stabilizernes $\mu$
\begin{equation}
    \sum_{i\in\mathcal A} m_i = 1 + 2\mu\sqrt N,\label{eq:facet_magic_polytope}
\end{equation}
where we have used that $\sum_{i\in \mathcal A} s_i = 1$. Define $C_N = 1 + 2\mu\sqrt N$. Note that if we have $3$ active dimensions, we obtain $C_N = 1 +2\sqrt{3}\mu $ which recovers the flat parallel surfaces to the stabilizer octahedron as shown in \cite{palhares_2026}.

Now we turn our attention to Eq.~\eqref{eq:m_grad}. For the active components, replace $\vec s$ via our displaced vector equation $\vec m=\vec s + 2\mu\mathbb{1}_N/\sqrt{N}$. We obtain
\begin{equation}
    \vec \Delta = 2\nu \vec m +\frac{4\mu\lambda}{\sqrt{N}}\mathbb{1}_N.\label{eq:equilibrium}
\end{equation}
Consider the complementary slackness conditions involving the multiplier $\lambda$:
\begin{equation}
    \lambda (\|\vec{m}-\vec{s}\|^2_{\mathcal{A},2}-4\mu^2) = 0.
\end{equation}
This condition splits the problem into two regimes. The first one occurs when $\mu$ is greater enough so the constraint is not active ($\|\vec{m}-\vec{s}\|^2_{\mathcal{A},2} <4\mu^2$) , in this case $\lambda =0$. 
From Eq.~(\ref{eq:equilibrium}) we have $ \vec{\Delta}= 2\nu\vec{m}$. However, $\vec{m}$ is still constrained to   $\|\vec m\|_{\mathcal{A},2} \leq 1$, so to maximize $P_\text{suc}^\mu$ we set $\|\vec m\|_{\mathcal{A},2} = 1$, obtaining $2\nu = \|\vec{\Delta}\|_{\mathcal{A},2}$. Therefore $\vec m = \vec\Delta/\|\vec \Delta\|_{\mathcal{A},2}$ and inserting it in the success probability we obtain,
\begin{align}
    P^\mathcal A_{\text{suc}} &= \frac{1}{2} + \frac{1}{4}\vec{\Delta} \cdot \vec{m}\nonumber\\
    &=\frac{1}{2} + \frac{1}{4}\|\vec \Delta\|_{\mathcal{A},2},
\end{align}
which is the Helstrom bound as expected. 

The second regime is whenever the constraint is active,  $\|\vec{m}-\vec{s}\|^2_{\mathcal{A},2}=4\mu^2$. In this case, we need to solve explicitly for the multipliers $\gamma$ and $\nu$. Summing each component in Eq.~\eqref{eq:equilibrium} and imposing the non-stabilizer constraint $\sum_{i\in\mathcal A} m_i = C_N$, we get
\begin{equation}
\frac{\|\vec\Delta\|_{\mathcal{A},1}  - N\gamma}{2\nu}=C_N ,
\end{equation}
where we have used $\gamma/\lambda = 4\mu/\sqrt{N}$ to express it in terms of $\gamma$. By imposing the projectivity constraint $\sum_{i\in\mathcal A} m_i^2 = 1$ we get 
\begin{equation}
    \frac{\|\vec\Delta\|^2_{\mathcal{A},2}-2\gamma\|\vec\Delta\|_{\mathcal{A},1}+N\gamma^2}{4\nu^2} = 1.
\end{equation}

Solving for $\gamma$  we get the quadratic equation
\begin{equation}
    N(C_N^2 - N)\gamma^2 - 2(C_N^2 - N)\|\vec{\Delta}\|_{\mathcal{A},1}\gamma + C_N^2\|\vec{\Delta}\|_{\mathcal{A},2}^2 - \|\vec \Delta\|_{\mathcal{A},1}^2 = 0,
\end{equation}
and by choosing the square root that leaves $\nu\geq 0$, we can insert the optimal multipliers to obtain

\begin{equation}
    \vec\Delta\cdot\vec m =\frac{\sqrt{(N\|\vec\Delta\|_{\mathcal{A},2}^2 - \|\vec\Delta\|_{\mathcal{A},1}^2)(N-C_N^2)} +C_N\|\vec\Delta\|_{\mathcal{A},1}}{N}. \label{eq:Optimal_overlap_with_mconstraint}
\end{equation}
Therefore, putting all together we get 
\begin{equation}\label{eq: Psucc_Nostab_m_qbits}
P^\mathcal{A}_{\mathrm{suc}}(\mu) = 
\begin{cases}
\frac{1}{2} + \frac{1}{4N}(\sqrt{(N\|\vec\Delta\|_{\mathcal{A},2}^2 - \|\vec\Delta\|_{\mathcal{A},1}^2)(N-C_N^2)} +C_N\|\vec\Delta\|_{\mathcal{A},1}), & \text{if } C_N < \frac{\|\vec\Delta\|_{\mathcal{A},1}}{\|\vec\Delta\|_{\mathcal{A},2}} \\
\frac{1}{2} + \frac{1}{4}\|\vec{\Delta}\|_{\mathcal{A},2}, & \text{if } C_N \ge \frac{\|\vec\Delta\|_{\mathcal{A},1}}{\|\vec\Delta\|_{\mathcal{A},2}}
\end{cases}.
\end{equation}

Note that to get the critical $C_N$ we used the fact that if Bob were to have an infinite amount of available non-stabilizerness, he will simply align his measurement with $\Delta$. In this case the optimal measurement vector $\vec m$ would correspond to the Helstrom measurement vector $m_i= \Delta_i/\|\vec \Delta\|_{\mathcal A,2}$. Following Eq.~(\ref{eq:facet_magic_polytope}), the amount of magic in this case is $C_N=\sum_{i\in\mathcal A} m_i = \|\vec\Delta\|_{\mathcal A,1}/\|\vec \Delta\|_{\mathcal A,2}$. Therefore, whenever $C_N<\|\vec\Delta\|_{\mathcal A,1}/\|\vec \Delta\|_{\mathcal A,2}$, the Helstrom bound is unattainable and Eq.~(\ref{eq:Optimal_overlap_with_mconstraint}) must be used. 
The active set $\mathcal A$ determines the dimension $N=\text{dim}(\mathcal A)$ of the face of the stabilizer polytope supporting the optimal solution. For qubit systems, $N\in\{1,2,3\}$, and the optimal success probability is obtained by selecting the admissible solution yielding the largest value of Eq.~\eqref{eq: Psucc_Nostab_m_qbits}. In the following we explain this in detail, showing how to obtain the maximum success probability  under the no-leakage assumption. 

First note that the above derivation was obtained assuming an active space $\mathcal A$. To obtain the maximum success probability for a fixed $\mu$ we must optimize over all phyisically realizable active spaces $\mathcal A$ (which could be a vertex, an edge or a face of the polytope). Therefore the solution for the (restricted to no-leakage)  optimization problem is
\begin{equation}
    \max_{\mathcal A}P^\mathcal{A}_{\mathrm{suc}}(\mu).\label{eq:P_suc_mu_sm}
\end{equation}

To determine which active spaces $\mathcal{A}$ admit a physical realization, we must ensure that the optimal state $\vec{s}$ remains strictly inside the $N$-dimensional facet. This requires the components of the $\vec s$ to remain non-negative: $s_i \ge 0$ for all $i \in \mathcal{A}$. We can determine the exact boundary where this condition fails by explicitly solving for $s_i$. From the gradient of the measurement vector in Eq.~(\ref{eq:m_grad}), we substitute the active space condition $2\lambda(m_i - s_i) = \gamma$ from Eq.~(\ref{eq:s_grad_active}) to obtain: 
\begin{equation}
    \Delta_i - 2\nu m_i - \gamma = 0 \implies m_i = \frac{\Delta_i - \gamma}{2\nu}.
\end{equation}
Furthermore, we previously established we can express the measurement components as a displacement from the closest stabilizer state 
\begin{equation}
    m_i = s_i + \frac{2\mu}{\sqrt{N}}.
\end{equation}
Equating these two expressions for $m_i$ allows us to isolate the state components $s_i$:
\begin{equation}
    s_i + \frac{2\mu}{\sqrt{N}} = \frac{\Delta_i - \gamma}{2\nu} \implies s_i = \frac{\Delta_i - \gamma - \frac{4\mu\nu}{\sqrt{N}}}{2\nu}.
\end{equation}
Because we selected the root corresponding to $\nu \ge 0$, the constraint $s_i \ge 0$ imposes conditions on the components of $\vec{\Delta}$. For this to hold for all $i \in \mathcal{A}$, it must be satisfied by the smallest component within the active set. Let us define this smallest component as $\Delta_{\mathrm{min}} = \min_{i \in \mathcal{A}} \Delta_i$. We thus obtain the following boundary threshold:
\begin{equation}
    \Delta_{\mathrm{min}} \ge \gamma + \frac{4\mu\nu}{\sqrt{N}}. \label{eq:KKT_threshold}
\end{equation}
If Eq.~(\ref{eq:KKT_threshold}) is violated, the optimal  state $\vec s$ has  collided with the boundary of the $N$-dimensional facet (i.e., $s_{\mathrm{min}} < 0$), indicating that the current $N$-subspace ansatz is geometrically illegal. In such cases, the active dimension must be reduced ($N \to N-1$) by dropping the smallest component from the active set $\mathcal{A}$. To test that our construction is sound we have built an SDP restricted to the same no-leakage assumption, and we have shown that the success probabilities are exactly those predicted by Eq.~\eqref{eq:P_suc_mu_sm}. We included this code in  Ref.~\cite{zamora2026code}. 

Finally, although Eq.~\eqref{eq:P_suc_mu_sm} constitutes a lower bound to the success probability  when leakage is allowed, we show in Fig.~\ref{figure:Psucc_mu} it is a tight bound, saturating the SDP maximum whenever the states under consideration result in a $\vec\Delta$ aligned with the symmetry axis of the stabilizer octahedron.  

To understand this, consider a difference vector $\vec{\Delta}$ aligned with an $N$-dimensional symmetry axis of the stabilizer octahedron (for example, the $T$-state direction where $\Delta_1 = \Delta_2 = \Delta_3$, or an edge bisector where $\Delta_1 = \Delta_2$). By symmetry, there exists an optimal closest stabilizer state $\vec{s}$ with the same symmetry, distributing its weight equally among the active dimensions. Consequently, all active components satisfy $s_i>0$, so the admissibility condition in Eq.~(\ref{eq:KKT_threshold}) remains satisfied throughout the optimization. As $\mu$ increases, $\vec{s}$ moves through the interior of the corresponding face without reaching its boundary, and the active set $\mathcal{A}$ therefore remains unchanged for all $\mu \in [0,\mu_{\max}]$. Since the active set never changes, allowing the measurement vector to acquire components outside $\mathcal{A}$ cannot improve the objective, and the no-leakage ansatz remains optimal. Hence, the $N$-subspace solution coincides with the unrestricted SDP optimum. 
\end{proof}

\subsection{Proof of Result~\ref{result: QRAC bound}} \label{appendix: QRAC bound}

\textit{\textbf{Result}.~\ref{result: QRAC bound}~
     Let us split Alice's bitstrings in two disjoint sets: $\mathbb X_0^{(j)} =  \{{x\in\{0,1\}^n} |x_j =0\}$ and $\mathbb X^{(j)}_1 =  \{{x\in\{0,1\}^n} |x_j =1\}$. {When restricted to only stabilizer measurements the maximum value possible of $P_g$, denoted $P_{g}^{\mathrm{STAB}}$,} is given by 
     \begin{equation}
         P_{g}^{\mathrm{STAB}} = \frac{1}{2} +\frac{1}{4n}\sum_{j\in [n]}\|\vec{\Delta}^{(j)}\|_\infty.
     \end{equation}
     where $\Delta^{(j)} = \sigma^{(j)}_0- \sigma^{(j)}_1=\frac{1}{2}\vec{\Delta}^{(j)}\cdot\vec{\sigma}$ with $\vec{\sigma}$ a vector of Pauli matrices and $\sigma^{(j)}_{i} = \frac{1}{2^{n-1}}\sum_{x\in\mathbb{X}^{(j)}_i}\rho_x$ effective states.}
\begin{proof}
In a quantum random access code, the task for Bob is to guess the $x_j$ bit of one of the $2^n$ Alice's bitstrings $\mathbf{x} = (x_1,\dots,x_n)$, by performing the measurement $M_j$ for some  $j\in[n]=\{1,\dots,n\} $ on the encoded state $\rho_\mathbf{x}$.  The average success probability  of measuring outcome $b=x_j$ given the state $\rho_\mathbf{x}$ and measurements $\{M_{b|j}\}_b$ can be written as~\cite{ambainis2009}
\begin{align}
    P_{g} &= \frac{1}{2^nn}\sum_{j\in[n]}\sum_{x\in\{0,1\}^n}p(b=x_j|x,j)\nonumber\\
    &= \frac{1}{2^nn}\sum_{j\in[n]}\sum_{x\in\{0,1\}^n}\text{Tr}(\rho_xM_{x_j|j}).
\end{align}
Note that we can split the bitstrings into two disjoint sets, $\mathbb X_0^{(j)} =  \{{x\in\{0,1\}^n} |x_j =0\}$ and $\mathbb X^{(j)}_1 =  \{{x\in\{0,1\}^n} |x_j =1\}$. Since the sum of a function over a set is equivalent to the sum over its disjoint subsets, we have:
\begin{equation}
    P_{g} = \frac{1}{2^nn}\sum_{j\in[n]}\big(\sum_{x\in \mathbb{X}^{(j)}_0}\text{Tr}(\rho_xM_{0|j})+\sum_{x\in \mathbb{X}^{(j)}_1}\text{Tr}(\rho_xM_{1|j})\big),
\end{equation}
and by defining the effective states 
\begin{equation}
    \sigma^{(j)}_{i} = \frac{1}{2^{n-1}}\sum_{x\in\mathbb{X}^{(j)}_i}\rho_x,
\end{equation}
we have 
\begin{equation}
     P_{g} = \frac{1}{2n}\sum_{j\in [n]}\text{Tr}(\sigma^{(j)}_0M_{0|j} + \sigma^{(j)}_1M_{1|j}),
\end{equation}
which can be interpreted as a collection of $n$ binary quantum state  discrimination problems, one for each measurement $j$. Furthermore, since $M_{0|j} + M_{1|j} =\mathbb{I}$ for each $j$ and $\textrm{Tr}(\sigma^{(j)}_1)=1$, then 
\begin{align}
     P_{g} &= \frac{1}{2} +\frac{1}{2n}\sum_{j\in [n]}\text{Tr}([\sigma^{(j)}_0- \sigma^{(j)}_1] M_{0|j} )\nonumber\\
     &= \frac{1}{2} +\frac{1}{2n}\sum_{j\in [n]}\text{Tr}(\Delta^{(j)} M_{0|j} ),
\end{align}
where $\Delta^{(j)} = \sigma^{(j)}_0- \sigma^{(j)}_1$. Let now $\Vec{\Delta} \coloneq \vec{n}_{\sigma^{(j)}_0}- \vec{n}_{\sigma^{(j)}_1}$ where $\vec{n}_{\sigma^{(j)}_0}$ and $\vec{n}_{\sigma^{(j)}_1}$ are the Bloch vectors of $\sigma^{(j)}_0$ and $\sigma^{(j)}_1$ respectively.  If we restrict to only stabilizer measurements, following the proof for Lemma~\ref{lemma: qubit fixed stab measurement no ancilla}, then for each $j$ the maximum of the trace  is
\begin{equation}
    \max_{M_{0|j}\in \textrm{STAB}}\text{Tr}(\Delta^{(j)} M_{0|j}) = \frac{1}{2}\|\vec{\Delta}^{(j)}\|_\infty,
\end{equation}
Therefore the maximum value of an $n\to1$ QRAC restricted to stabilizer measurements is 
\begin{equation}
   P_{g}^\textrm{STAB} = \frac{1}{2} +\frac{1}{4n}\sum_{j\in [n]}\|\vec{\Delta}^{(j)}\|_\infty.
\end{equation}
\end{proof}

\subsection{Proof of Result~\ref{result: simulation fidelity}}\label{appendix: result simulation fidelity}

\textit{\textbf{Result}~\ref{result: simulation fidelity} Let $V,U$ be qubit unitaries where \hbox{$U = C_1 T C_2 : C_1, C_2 \in \mathcal{C}_1$}. Then \
     \begin{equation}
         F(U,V) \leq P^{\rm stab, a}_{\rm suc}(\Delta_V, \ketbra{H}) 
     \end{equation}
     where $\Delta_V = V^\dagger(\ketbra{0}-\ketbra{1})V$. }
\begin{proof}
    Firstly, let 
    \begin{equation}
        U V^\dagger = \begin{bmatrix}
            w_{00} & w_{01} \\
            w_{10} & w_{11}
        \end{bmatrix}.
    \end{equation}
    Then
    \begin{align}
        F(U,V) &= \frac{1}{4} \vert \textrm{tr}\big[ U V^\dagger \big] \vert^2, \\
        &= \frac{1}{4} \vert w_{00} + w_{11} \vert^2 \\
        &\leq \frac{1}{4} \bigg[ \vert w_{00} \vert + \vert w_{11} \vert \bigg]^2 \\
        &= \frac{1}{4} \bigg[ \vert w_{00} \vert^2 + \vert w_{11} \vert^2 + 2 \vert w_{00} \vert~ \vert w_{11} \vert \bigg] \\
        &\leq \frac{1}{2}  \bigg[ \vert w_{00} \vert^2 + \vert w_{11} \vert^2 \bigg]  ~ ~ {\rm as} ~ ~ 2 \vert  a \vert~ \vert b \vert \leq \vert a \vert^2 + \vert b \vert^2.  ~ 
    \end{align}
    Now, consider performing QSD for $\Delta_V = V^\dagger \ketbra{0} V - V^\dagger \ketbra{1} V$ using the POVM $\big\{ U^\dagger \ketbra{0} U, U^\dagger \ketbra{1} U \big\}$. 
    This POVM can be physically implemented by applying a unitary $U$ before measuring in the computational basis. 
    The success probability here is then  
    \begin{align}
        P_{\rm suc}(U, V) &= \frac{1}{2}\textrm{tr}\big[ \big(U^\dagger \ketbra{0} U \big) \big(V^\dagger \ketbra{0} V) \big] + \frac{1}{2}\textrm{tr}\big[ \big(U^\dagger \ketbra{1} U \big) \big(V^\dagger \ketbra{1} V) \big] \\
        &= \frac{1}{2} \big[ \vert \bra{0} U V^\dagger \ket{0} \vert^2 + \vert \bra{1} U V^\dagger \ket{1} \vert^2 \big] \\
        &=  \frac{1}{2}  \bigg[ \vert w_{00} \vert^2 + \vert w_{11} \vert^2 \bigg],
    \end{align}
    meaning that $F(U,V) \leq P_{\rm suc}(U,V)$. Note, this is independent of any measurement restrictions and includes a potentially suboptimal assignment of measurement outcomes to states.  
    
    Now, assume that a QSD player is restricted to perform only stabilizer circuits, but has access to a single $H$-type magic state. Using adaptive stabilizer circuits, this $H$-type magic states can be converted into a single $T$ gate (magic state injection), where 
    \begin{equation}
        T = \begin{bmatrix}
            1 & 0 \\
            0 & e^{\frac{i \pi}{4}}
        \end{bmatrix}
    \end{equation}
    The players is therefore able to measure any POVM of the form $\big\{ U \ketbra{0} U^\dagger, U \ketbra{1} U^\dagger \big\}$ where $U \in \mathcal{C}_1$ or $U=C_1 T C_2$, where $C_1, C_2 \in \mathcal{C}_1$ i.e., $U$ is a unitary that contains at most one $T$ gate in its decomposition into the universal gate set Clifford + $T$. Therefore, as this is a viable measurement strategy, and Result~\ref{result: appending any qubit state with feedfoward} optimises over all possible measurement strategies, it holds that when performing QSD on $\Delta_V = V^\dagger \ketbra{0} V - V^\dagger \ketbra{1} V $ that  $P(U, V) \leq P^{\rm stab, a}_{\rm suc}(\Delta_V, \ketbra{H})$. It therefore follows that $F(U,V) \leq P^{\rm stab, a}_{\rm suc}(\Delta_V, \ketbra{H})$. The only assumption about $U$ is that it contained at most one $T$ gate in its decomposition into Clifford + $T$. Moreover, the right hand side of the inequality is independent of $U$. Hence, it holds that the fidelity of a unitary $U$ --- that contains a single $T$ state --- and any arbitrary unitary $V$, is upper bounded by ones ability to discriminate between $V^\dagger \ketbra{0} V$ and $V^\dagger \ketbra{1} V$ when using adaptive stabilizer circuits with access to a single $H$-type magic state. Note, this bound is tight if $V$, and therefore $V^\dagger$, contains only a single $T$ gate in its decomposition into Clifford + $T$. 
\end{proof}

\subsection{Proof of Lemma~\ref{lemma: general state discrimination}}\label{sm: proof of general state discrimination}
Consider performing QSD on the states $\rho, \sigma \in \mathcal{D}(\mathbb{C}^d)$ using the POVM $\{ M_i \}_{i=1}^{N}$, such that $M_i \geq 0~\forall~i$ and $\sum_{i=1}^N M_i = \mathbb{I}$. The following Lemma allows the success probability for perform QSD under the optimal coarse graining to be calculated. 
\\
\\
{\em {\bf Lemma}~\ref{lemma: general state discrimination} If using a POVM $\{M_i \}_{i=1}^N$ for perform QSD, the maximum success probability under all possible coarse grainings is
    \begin{equation}
        P_{\rm suc}(\Delta) = \frac{1}{2} + \frac{1}{4} \sum_{i=1}^N \vert \textnormal{tr}\big[ M_i \Delta \big] \vert.
    \end{equation}
}
\begin{proof}
    Each POVM element in $\{M_i\}_{i=1}^N$ will be assigned to either $\rho$ or $\sigma$ such that upon receiving the output associated to that POVM element the player will say they have the state $\rho$ or $\sigma$ respectively. The POVM element $M_j \in \{M_i\}_{i=1}^N$ will be assigned to the output $\rho$ if 
    \begin{equation}
        \textrm{tr}\big[M_j \rho \big] \geq \textrm{tr}\big[M_j \sigma \big],
    \end{equation}
    meaning that the probability of getting the outcome associated to $M_j$ is greater for $\rho$ then $\sigma$. This can equivalently be written as 
    \begin{equation}
        \textrm{tr}\big[ M_i \Delta \big] \geq 0,
    \end{equation}
    where $\Delta \coloneq \rho - \sigma$. The $N$ element POVM can then be coarse grained into two POVM elements:
    \begin{equation}
        M_\rho = \sum_i M_i : \textrm{tr}\big[ M_i \Delta \big] \geq 0 ~\forall~i
    \end{equation}
    and 
    \begin{equation}
        M_\sigma = \mathbb{I} - M_\rho,
    \end{equation}
    where the player outputs $\rho (\sigma)$ if they get the outcome associated to $M_\rho (M_\sigma)$. The POVM element $M_\rho$ can be summarised as 
    \begin{equation}
        M_\rho = \sum_i M_i S_i, \label{eq: coarse grained POVM}
    \end{equation}
    where 
    \begin{equation}
        S_i = \begin{cases}
            1 &\textrm{if}~ \textrm{tr}\big[ M_i \Delta \big] \geq 0 \\
            0 &\textrm{if}~ \textrm{tr}\big[ M_i \Delta \big] < 0. \label{eq: s equation}
        \end{cases}
    \end{equation}
    Now, using the two outcome POVM $\{ M_\rho, \mathbb{I}-M_\rho \}$, the success probability for QSD using $\{ M_i \}_{i=1}^N$ is given by 
    \begin{equation}
        \begin{split}
            P_{\rm suc}(\Delta) &= \frac{1}{2} \textrm{tr}\big[ M_\rho \rho \big] + \frac{1}{2} \textrm{tr} \big[ (\mathbb{I}-M_\rho) \sigma \big] \\
            &=\frac{1}{2} + \frac{1}{2} \textrm{tr}\big[M_\rho \Delta \big].
        \end{split}
    \end{equation}
    By inputting Eq.\eqref{eq: coarse grained POVM}, this becomes
    \begin{equation}
        P_{\rm suc}(\Delta) = \frac{1}{2} + \frac{1}{2} \sum_j \textrm{tr}\big[M_i\Delta \big]S_i.
    \end{equation}
    Now, using Eq.\eqref{eq: s equation}, it can be seen that 
    \begin{equation}
        \textrm{tr}\big[M_j \Delta\big] S_j = \max \big\{ \textrm{tr}\big[ M_j \Delta \big], 0 \big\}.
    \end{equation}
    Then, using the fact that 
    \begin{equation}
        \max \{ a, b \} = \frac{ a + b + \vert a - b \vert}{2},
    \end{equation}
    the success probability can be written as
    \begin{align}
        P_{\rm suc}(\Delta) &= \frac{1}{2} + \frac{1}{2}~\sum_i \frac{\vert \textrm{tr}\big[ M_i \Delta \big] \vert + \textrm{tr}\big[ M_i \Delta]}{2} \\
        &= \frac{1}{2} + \frac{1}{4}~ \sum_i \vert \textrm{tr}\big[ M_i \Delta \big] \vert + \frac{1}{4} \sum_i \textrm{tr}\big[ M_i \Delta \big] \\
        &=  \frac{1}{2} + \frac{1}{4}~ \sum_i \vert \textrm{tr}\big[ M_i \Delta \big] \vert + \frac{1}{4} \textrm{tr}\bigg[ \big( \sum_i M_i \big) \Delta \bigg] \\
        &= \frac{1}{2} + \frac{1}{4}~ \sum_i \vert \textrm{tr}\big[ M_i \Delta \big]  \vert + \frac{1}{4} \textrm{tr}\big[ \Delta \big] \\
        &= \frac{1}{2} + \frac{1}{4}~  \sum_i \vert \textrm{tr}\big[ M_i \Delta \big] \vert,
    \end{align}
    where we have used the fact that $\sum_i M_i = \mathbb{I}$ and $\textrm{tr}\big[ \Delta \big] = 0$. This completes the proof.
\end{proof}

\section{QSD under Adaptive Stabilizer Circuits}\label{SM: counter example}
Here, we show an example of a QSD tasks with adaptive stabilizer circuits in which a stabilizer state in an ancilla allows the minimum error success probability to be improved as compared to fixed stabilizer circuits. 

Consider that the referee gives one of the following two two-qubit states with equal probability to the player: 
\begin{equation}
    \begin{split}
        \rho &= \frac{1}{2}(\ketbra{00}_{A_1A_2} + \ketbra{1+}_{A_1A_2}) \\
        \sigma &= \frac{1}{2}(\ketbra{01}_{A_1A_2} + \ketbra{1-}_{A_1A_2}),
    \end{split}
\end{equation}
where $\rho, \sigma \in \mathcal{D}(\mathcal{H}^{A_1}_1 \otimes \mathcal{H}^{A_2}_1)$. Hence,

\begin{equation}
    \begin{split}
        \Delta &= \rho - \sigma \\
        &= \frac{1}{2} \bigg( \ketbra{0} \otimes Z + \ketbra{1} \otimes X \bigg) \\
        &= \frac{1}{4} \bigg( \mathbb{I}_{A_1}Z_{A_2} + {Z}_{A_1}Z_{A_2} + \mathbb{I}_{A_1}X_{A_2} - {Z}_{A_1}X_{A_2} \bigg).
    \end{split}
\end{equation}
Now, consider appending $\ketbra{0} \in {\rm STAB}(1)$ in a space $\mathcal{H}_1^B$, and then performing $C_{A_1}\mathbb{I}_{A_2}X_B$, i.e., controlled-X between $A_1$ and $B$, before measuring $B$ in the computational basis. If the outcome of this computational basis measurement is $0$, then apply $\mathbb{I}_{A_1A_2}$; if the outcome is $1$, apply $\mathbb{I}_{A_1}H_{A_2}$. Finally, measure $A_1$ and $A_2$ in the computational basis. 

The effective POVM elements of such an adaptive circuit are of the form 
\begin{equation}
    \big\{ T^{U, \omega, \{V\}}_{a,b} = \textrm{tr}_B\big[ U^\dagger \big( V_b^\dagger \ketbra{a} V_b \otimes \ketbra{b} \big) U (\mathbb{I}_A \otimes \omega_B) \big] \big\}_{a,b}~,
\end{equation}
which in this case are specifically 
\begin{equation}
    \begin{split}
        T^{C_{A_1}\mathbb{I}_{A_2}X_B, ~\ketbra{0}, ~\{\mathbb{I}_{A_1A_2}, \mathbb{I}_{A_1}H_{A_2}\}}_{00,0} &= \ketbra{00} \\
        T^{C_{A_1}\mathbb{I}_{A_2}X_B, ~\ketbra{0}, ~\{\mathbb{I}_{A_1A_2}, \mathbb{I}_{A_1}H_{A_2}\}}_{01,0} &= \ketbra{01} \\
        T^{C_{A_1}\mathbb{I}_{A_2}X_B, ~\ketbra{0}, ~\{\mathbb{I}_{A_1A_2}, \mathbb{I}_{A_1}H_{A_2}\}}_{10,0} &= 0 \\
        T^{C_{A_1}\mathbb{I}_{A_2}X_B, ~\ketbra{0}, ~\{\mathbb{I}_{A_1A_2}, \mathbb{I}_{A_1}H_{A_2}\}}_{11,0} &= 0 \\
        T^{C_{A_1}\mathbb{I}_{A_2}X_B, ~\ketbra{0}, ~\{\mathbb{I}_{A_1A_2}, \mathbb{I}_{A_1}H_{A_2}\}}_{00,1} &= 0 \\
        T^{C_{A_1}\mathbb{I}_{A_2}X_B, ~\ketbra{0}, ~\{\mathbb{I}_{A_1A_2}, \mathbb{I}_{A_1}H_{A_2}\}}_{01,1} &= 0 \\
        T^{C_{A_1}\mathbb{I}_{A_2}X_B, ~\ketbra{0}, ~\{\mathbb{I}_{A_1A_2}, \mathbb{I}_{A_1}H_{A_2}\}}_{10,1} &= \ketbra{1+} \\
        T^{C_{A_1}\mathbb{I}_{A_2}X_B, ~\ketbra{0}, ~\{\mathbb{I}_{A_1A_2}, \mathbb{I}_{A_1}H_{A_2}\}}_{11,1} &= \ketbra{1-}, \\
    \end{split}
\end{equation}
meaning that this circuit will lead to a total effective POVM of the form
\begin{equation}
    \big\{ \ketbra{00}, \ketbra{01}, \ketbra{1+}, \ketbra{1-} \big\},
\end{equation}
being measured. This is a basis of stabilizer states, but is not a stabilizer basis. Employing Lemma~\ref{lemma: general state discrimination}, this POVM can then be seen to be able to perfectly distinguish between $\rho$ and $\sigma$, such that $P^{\rm stab, a}_{\rm suc}(\Delta, \ketbra{0})=1$. 

Now, employing Lemma~\ref{lemma: appendix : two qubit min error sucess prob} and using the Pauli decomposition of $\Delta$ given above, it can be seen that 
\begin{equation}
    \begin{split}
        P^{\rm stab}_{\rm suc}(\Delta) &= \frac{1}{2} + \frac{1}{2} \times \bigg(2 \times \frac{1}{4}\bigg), \\
        &= \frac{3}{4}.
    \end{split}
\end{equation}
Hence, $P^{\rm stab}_{\rm suc}(\Delta) < P^{\rm stab, a}_{\rm suc}(\Delta, \ketbra{0})$, meaning that there exists an example where $P^{\rm stab}_{\rm suc}(\Delta) < P^{\rm stab, a}_{\rm suc}(\Delta, \omega)$ with $\omega \in {\rm STAB}$. Therefore, under adaptive stabilizer circuits even access to stabilizer ancilla can improve the minimum error success probability, meaning that Result~\ref{theorem: STABs don't increase sucess prob} does not in general extend to adaptive stabilizer circuits.

As a final point, we note that if one measured $Z$ on the first qubit, and then measured $Z$ on the second if the output from the first measurement was $0$, and $X$ if the output from the first measurement was $1$, then they would achieve unity in this QSD task. However, we here considered all adaptivity to be in the ancillary system, and not within the primary system. Hence, this falls outside our considered framework, and is left for future work.

\end{document}